%% file: ourmieres-vega.tex
\documentclass[11pt,a4paper]{amsart}

\usepackage{euler}
\usepackage{bbm,amsfonts,amsmath,amssymb,mathrsfs,tikz,hyperref,dsfont,
csquotes,enumerate,subfigure} 
\hypersetup{pdfborder={0000}, colorlinks=true, linkcolor=blue,citecolor=citegreen}
\definecolor{citegreen}{rgb}{0.2,0.2,0.6}
\usepackage[%babel=true,
kerning=true]{microtype}

\hypersetup{
 colorlinks=true,
 citecolor=darkred,
 linkcolor=blue,
 urlcolor=blue}

\usepackage[shortlabels]{enumitem}

\usepackage[font=small]{caption}

\usepackage{palatino}
\usepackage{eucal}
\usepackage{verbatim}

\usepackage{amsmath,amssymb,amsthm,graphicx,epstopdf,mathrsfs,url}

\hypersetup{
 colorlinks=true,
 citecolor=darkred,
 linkcolor=blue,
 urlcolor=blue}

\input commands.tex
\author{$\text{Thomas Ourmi\`eres-Bonafos}^{1}$}
\address[1]{Laboratoire de Math\'ematiques d'Orsay, Univ.~Paris-Sud, CNRS, Universit\'e Paris-Saclay, 91405 Orsay, France}

\email{thomas.ourmieres-bonafos@math.u-psud.fr}
\urladdr{http://www.math.u-psud.fr/~ourmieres-bonafos/}

\author{$\text{Luis Vega}^{2,3}$}
\address[2]{BCAM -Basque Center for Applied Mathematics, Alameda de Mazarredo, 14 E48009 Bilbao, Basque Country -  Spain}
\address[3]{Departamento de Matem\'aticas - Universidad del Pa\'is Vasco UPV-EHU, Apdo 644, 48080  Bilbao, Spain}
\email{lvega@bcamath.org}
\urladdr{http://www.bcamath.org/en/people/lvega}

\title[\textsc{Self-adjointness of Dirac operators}]{\textsc{A strategy for self-adjointness of Dirac operators : applications to the MIT bag model and $\delta$-shell interactions}}
\date{}
\begin{document}
\maketitle

\begin{abstract} We develop an approach to prove self-adjointness of Dirac operators with boundary or transmission conditions at a $\mathcal{C}^2$-compact surface without boundary. To do so we are lead to study the layer potential induced by the Dirac system as well as  to define traces in a weak sense for functions in the appropriate Sobolev space. Finally, we introduce Calder\'on projectors associated with the problem and illustrate the method in two special cases: the well-known MIT bag model and an electrostatic $\delta$-shell interaction.
\end{abstract}

\tableofcontents
\section{Introduction}
\subsection{Motivations} The study of massive relativistic particles of spin-$1/2$ such as electrons or quarks involves the Dirac operator and such systems are of great importance in elementary particle physics.

From a mathematical physics point of view, this operator attracted a lot of attention in the past few years and the first step in its study is the understanding of either boundary or transmission conditions through a surface in order to prove self-adjointness of the Dirac operator with such conditions.

We aim to develop a general strategy involving boundary integral operators and associated Calder\'on projectors and we apply this method to the well-known MIT bag model and the Dirac operators coupled with an electrostatic $\delta$-shell interaction.

The MIT bag model is used to study confined particles of spin-$1/2$ into domains of $\mathbb{R}^3$ (for more physical motivations, see \cite{Cho75,Cho74,Cho74-2,Jo75}). This system has recently been studied in \cite{ALTR16} and, in particular, the authors prove self-adjointness of the operator for smooth domains.

The Dirac operator coupled with a $\delta$-interaction also attracted a lot of attention in the past few years. To our knowledge, the first paper dealing with this question is \cite{DES89} where the authors study the particular case of an interaction supported on a sphere. They take advantage of the symmetry of the system in order to answer the question of self-adjointness and study spectral properties of the system. In the sequence of papers \cite{AMV14,AMV15,AMV16} the question of self-adjointness is handled for $\mathcal{C}^2$-surfaces and spectral properties are investigated. Recently, a strategy using quasi-boundary triplets was exposed in \cite{BEHL16} in order to study this system and the authors recover and extend some results of \cite{AMV14} for $\mathcal{C}^\infty$-smooth surfaces. However, in these works, the authors fail to prove self-adjointness for critical values of the coupling constant. Our initial motivation here was to understand this phenomenom and we prove that for these critical values the operator is self-adjoint on a larger domain. It is worth mentionning that simultaneously, in \cite{BH16}, the authors recover similar results with a boundary triplet technique. This phenomenom is reminescent of similar questions in the context of negative-index materials investigated in \cite{BDR99,CKP16}.

This paper is inspired by the strategy developped in \cite{BFSVDB16} about two-dimensional Dirac operators with graphene boundary conditions. However, in our case, the situation is more involved because we study a layer potential on a general $\mathcal{C}^2$-surface which, in dimension two, was done studying the twin of this layer potential on the circle and extending the results by the Riemann mapping theorem. Hence, in the present article, we introduce a framework for boundary integral operators in the same spirit as the one developed, for instance in \cite{Cos88}, for elliptic operators of order two. Indeed, we study the layer potential operator associated with the Dirac problem and study regularisation properties. To do so, we are led to introduce and study various properties about the natural Sobolev space associated with the Dirac operator. It allows to define boundary values of such a layer potential in a weak sense.

Finally, we mention that the layer potential for the Dirac operator was studied in \cite{AGHMC01} for Lipschitz hypersurfaces. As here we ask for $\mathcal{C}^2$-regularity of the surface, it allows to define weaker data on the boundary. It is also worth mentionning that our strategy share similarities with the work exposed in \cite{BLC09} about $\mathcal{C}^\infty$-boundary techniques and pseudo-differential tools.

\subsection{Notations and definitions}
Before going any further we need to introduce a few notation and definition.

\subsubsection{Basic notations}

The set $\mathbb{N} = \{0,1,2,\dots\}$ denotes the set of natural integers and we define $\mathbb{N}^* = \mathbb{N}\setminus\{0\}$.  $\mathbb{R}$ and $\mathbb{C}$ are the fields of real and complex numbers, respectively. For $z\in\mathbb{C}$, $\overline{z}$ is its conjugate.

Let $d\in\mathbb{N}^*$, $x=(x_1,\cdots,x_d)$ denote the cartesian coordinates in the euclidean space $\mathbb{R}^d$ and $\mathbf{0}$ the origin.

For a Hilbert space $\mathfrak{H}$, $\ps{\cdot}{\cdot}_{\mathfrak{H}}$ and $\|\cdot\|_{\mathfrak{H}}$ denote the scalar product and the norm on $\mathfrak{H}$, respectively. When $\mathfrak{H} =\mathbb{C}^d$ the scalar product $\ps{\cdot}{\cdot}_{\mathbb{C}^d}$ is taken antilinear with respect to the second variable.

For a matrix $A=(A_{i,j})_{i,j\in\{1,\dots,d\}}\in \mathbb{C}^{d\times d}$, $A^*$ is the conjugate transpose of $A$, that is $(A_{i,j}^*)_{i,j\in\{1,\dots,d\}}=(\overline{A}_{j,i})_{i,j\in\{1,\dots,d\}}$. For any $z_1,z_2\in\mathbb{C}^d$, it satisfies $\ps{Az_1}{z_2}_{\mathbb{C}^d} = \ps{z_1}{A^*z_2}_{\mathbb{C}^d}$. We also introduce $\|\cdot\|_{\mathcal{M}}$ the matricial norm defined as $\|A\|_{\mathcal{M}} = \sup_{\|z\|_{\mathbb{C}^d} = 1}\|A z\|_{\mathbb{C}^d}$. The identity of $\mathbb{C}^{d\times d}$ will be denoted $\rm{Id}$.

For a metric space $\mathsf{X}$, $\mathsf{X}'$ denotes its topological dual and$\ps{\cdot}{\cdot}_{\mathsf{X}',\mathsf{X}}$ the duality pairing between $\mathsf{X}'$ and $\mathsf{X}$. Let $\mathsf{X}$ and $\mathsf{Y}$ be two metric spaces and $\mathcal{L}$ a bounded linear operator from $\mathsf{X}$ to $\mathsf{Y}$. $\mathcal{L}'$ denote its adjoint and we recall that it is a bounded linear operator from $\mathsf{Y}'$ to $\mathsf{X}'$.

Let $\mathcal{U} \subset \mathbb{R}^d$. Its closure in $\mathbb{R}^d$ is denoted $\overline{\mathcal{U}}$. We also introduce the open ball of center the origin $\mathbf{0}$ and radius $R>0$ as $B(R) := \{ x\in\mathbb{R}^d : \|x\|_{\mathbb{R}^d} < R\}$. $\mathsf{dist}(x,K)$ denotes the distance of a point $x\in\mathbb{R}^d$ to a compact subset $K\subset \mathbb{R}^d$.\\

From now on $p\in\mathbb{N}^*$ and when $p=1$, the mention $p$ is dropped in the following notation. Let $\Omega$ be a $\mathcal{C}^2$-domain of $\mathbb{R}^d$. If the boundary $\partial\Omega$ of $\Omega$ is non-empty, we set $\Sigma:=\partial\Omega$ and denote by $\bfn$ its outward pointing normal and $\dd\mathfrak{s}$ the $(d-1)$-dimensional Hausdorff measure on $\Sigma$. We assume that $\Sigma$ is compact, connected and without boundary.

\subsubsection{Spaces of smooth functions and distributions}
$\mathcal{C}^\infty(\Omega)^p$ denotes the usual space of indefinitely differentiable functions with values in $\CC^p$. Similarly, $\mathcal{C}_0^\infty(\Omega)^p$ is the set of indefinitely differentiable functions with values in $\CC^p$ with compact support. If $\Omega$ is bounded, the space $\mathcal{C}_0^\infty(\overline{\Omega})^p$ can be identified with $\mathcal{C}^\infty(\overline{\Omega})^p$. $\mathcal{C}_0^\infty(\Omega)^p$ can also be denoted $\mathcal{D}(\Omega)^p$ and endowed with its usual family of semi-norms it is a metric space. The space of distributions is defined as $\mathcal{D}'(\Omega)^p = \big(\mathcal{C}_0^\infty(\Omega)^p \big )'$. For $u\in\mathcal{D}'(\Omega)$, $\mathsf{supp}(u)$ denotes the support of the distribution $u$.

The Schwarz class $\mathcal{S}(\mathbb{R}^d)^p$ is defined as
%%%%%%%%%%%%%%%%%%%%%%%%%%%%%%%%%%%%%%%%%%%%%%%%%%%%%%%%
\[
	\mathcal{S}(\mathbb{R}^d)^p := \{f\in\mathcal{C}^\infty(\mathbb{R}^d)^p : \text{for all } (k,l)\in\mathbb{N}^d\times\mathbb{N}^d, \sup_{x\in\mathbb{R}^d} \|x^k \partial^l f(x)\|_{\mathbb{C}^p}<+\infty\},
\]
%%%%%%%%%%%%%%%%%%%%%%%%%%%%%%%%%%%%%%%%%%%%%%%%%%%%%%%%
where we used the multi-index notation. More precisely if $k=(k_1,\dots,k_d)\in \mathbb{N}^d$ and $x\in\mathbb{R}^d$, $x^k =x_1^{k_1}\dots x_d^{k_d}\in\mathbb{R}$ and $\partial^k=\partial_1^{k_1}\dots\partial_d^{k_d}$. Endowed with its usual family of semi-norms, $\mathcal{S}(\mathbb{R}^d)^p$ is a metric space and the space of tempered distributions is defined as $\mathcal{S}'(\mathbb{R}^d)^p = \big(\mathcal{S}(\mathbb{R}^d)^p\big)'$.

\subsubsection{$L^q$-spaces}
Let $q\geq 1$. $L^q(\Omega)^p$ is the space of functions $f$, which are measurable with respect to the Lebesgue measure and with values in $\CC^p$, such that
%%%%%%%%%%%%%%%%%%%%%%%%%%%%%%%%%%%%%%%%%%%%%%%%%%%%%%%%
\[
	\|f\|_{L^q(\Omega)^p}^q := \int_{\mathbb{R}^d}\|f\|_{\CC^p}^q\dd x < +\infty.
\]
%%%%%%%%%%%%%%%%%%%%%%%%%%%%%%%%%%%%%%%%%%%%%%%%%%%%%%%%
When $q=2$, $L^2(\Omega)^p$ is a Hilbert space and its scalar product is given by
%%%%%%%%%%%%%%%%%%%%%%%%%%%%%%%%%%%%%%%%%%%%%%%%%%%%%%%%
\[
	\ps{f}{g}_{L^2(\Omega)^p} = \int_{\Omega}\ps{f(x)}{g(x)}_{\CC^p}\dd x,\quad f,g\in L^2(\Omega)^p.
\]
%%%%%%%%%%%%%%%%%%%%%%%%%%%%%%%%%%%%%%%%%%%%%%%%%%%%%%%%
If $f\in L^2(\mathbb{R}^d)^p$, $f|_\mathcal{\Omega}$ denotes the restriction of $f$ to the domain $\Omega$.

We also introduce the space $L^\infty(\Omega)^p$ as the space of bounded $\CC^p$-valued functions. For $f\in L^\infty(\Omega)^p$, the associated norm is defined as
%%%%%%%%%%%%%%%%%%%%%%%%%%%%%%%%%%%%%%%%%%%%%%%%%%%%%%%%
\[	
	\|f\|_{L^\infty(\Omega)^p} = \sup_{x\in\Omega} \|f(x)\|_{\CC^p}.
\]
%%%%%%%%%%%%%%%%%%%%%%%%%%%%%%%%%%%%%%%%%%%%%%%%%%%%%%%%

\subsubsection{Fourier transform}
\label{subsub:fourtrans}
For a function $f\in L^1(\mathbb{R}^d)^p$, we introduce its Fourier transform as
%%%%%%%%%%%%%%%%%%%%%%%%%%%%%%%%%%%%%%%%%%%%%%%%%%%%%%%%
\[
	\mathcal{F}(f)(\xi) = \frac{1}{(2\pi)^{d/2}}\int_{\mathbb{R}^d}e^{-i\ps{x}{\xi}_{\mathbb{R}^d}} f(x) \dd x\in\mathbb{C}^p,\quad \text{for all } \xi\in\mathbb{R}^d.
\]
%%%%%%%%%%%%%%%%%%%%%%%%%%%%%%%%%%%%%%%%%%%%%%%%%%%%%%%%
The Fourier transform can be extended into an isometry of $L^2(\mathbb{R}^d)^p$ and it is well known that $\mathcal{F}$, seen as an operator from $\mathcal{S}(\mathbb{R}^d)^p$ to $\mathcal{S}(\mathbb{R}^d)^p$ is invertible and the inverse Fourier transform $\mathcal{F}^{-1}$ is given by
%%%%%%%%%%%%%%%%%%%%%%%%%%%%%%%%%%%%%%%%%%%%%%%%%%%%%%%%
\[
	\mathcal{F}^{-1}(f)(x) = \frac{1}{(2\pi)^{d/2}}\int_{\mathbb{R}^d}e^{i\ps{\xi}{x}_{\mathbb{R}^d}} f(\xi) \dd \xi\in\mathbb{C}^p,\quad \text{for all } f\in\mathcal{S}(\mathbb{R}^d)^p \text{ and }x\in\mathbb{R}^d.
\]
%%%%%%%%%%%%%%%%%%%%%%%%%%%%%%%%%%%%%%%%%%%%%%%%%%%%%%%%
By duality, we can also extend $\mathcal{F}$ to the space of tempered distributions $\mathcal{S}'(\mathbb{R}^d)^p$.

\subsubsection{Sobolev spaces}
Let $|s|\leq1$, we introduce the usual Sobolev space $H^s(\mathbb{R}^d)^p$ as:
%%%%%%%%%%%%%%%%%%%%%%%%%%%%%%%%%%%%%%%%%%%%%%%%%%%%%%%%
\[
	H^s(\mathbb{R}^d)^p := \{f\in L^2(\mathbb{R}^d)^p : 	\|f\|_{H^s(\mathbb{R}^d)^p} < +\infty\},
\]
%%%%%%%%%%%%%%%%%%%%%%%%%%%%%%%%%%%%%%%%%%%%%%%%%%%%%%%%
where
%%%%%%%%%%%%%%%%%%%%%%%%%%%%%%%%%%%%%%%%%%%%%%%%%%%%%%%%
\[
	\|f\|_{H^s(\mathbb{R}^d)^p}^2 = \int_{\xi\in\mathbb{R}^d}(1+|\xi|^2)^{s}\|\mathcal{F}(f)(\xi)\|_{\CC^p}^2\dd \xi.
\]
%%%%%%%%%%%%%%%%%%%%%%%%%%%%%%%%%%%%%%%%%%%%%%%%%%%%%%%%
If $s=0$, as the Fourier transform is unitary on $L^2(\mathbb{R}^d)^p$, we recover by definition $H^0(\mathbb{R}^d)^p = L^2(\mathbb{R}^d)^p$.

The space $H^s(\Omega)^p$ is defined as follows
%%%%%%%%%%%%%%%%%%%%%%%%%%%%%%%%%%%%%%%%%%%%%%%%%%%%%%%%
\[
	H^s(\Omega)^p = \{f \in L^2(\Omega)^p : \text{there exists } \widetilde{f}\in H^s(\mathbb{R}^d)^p\text{ such that }\widetilde{f}|_\Omega = f\}.
\]
%%%%%%%%%%%%%%%%%%%%%%%%%%%%%%%%%%%%%%%%%%%%%%%%%%%%%%%%
and for $f\in H^s(\Omega)^p$ the associated norm is given by
%%%%%%%%%%%%%%%%%%%%%%%%%%%%%%%%%%%%%%%%%%%%%%%%%%%%%%%%
\[
	\|f\|_{H^s(\Omega)^p} = \inf_{\widetilde{f}\in H^s(\mathbb{R}^d)^p, \widetilde{f}|_\Omega=f} \|\widetilde{f}\|_{H^s(\mathbb{R}^d)^p}.
\]
%%%%%%%%%%%%%%%%%%%%%%%%%%%%%%%%%%%%%%%%%%%%%%%%%%%%%%%%
Now, let and $\mathfrak{H}^1(\Omega)^p$ be the space
%%%%%%%%%%%%%%%%%%%%%%%%%%%%%%%%%%%%%%%%%%%%%%%%%%%%%%%%
\[
	\mathfrak{H}^1(\Omega)^p := \{f = (f_j)_{j\in\{1,\dots,p\}}\in L^2(\Omega)^p : \text{for all } j\in\{1,\dots,p\}, \nabla f_j \in L^2(\Omega)^p \},
\]
%%%%%%%%%%%%%%%%%%%%%%%%%%%%%%%%%%%%%%%%%%%%%%%%%%%%%%%%
associated with the norm
%%%%%%%%%%%%%%%%%%%%%%%%%%%%%%%%%%%%%%%%%%%%%%%%%%%%%%%%
\[
	\|f\|_{\mathfrak{H}^1(\Omega)^p}^2 = \|f\|_{L^2(\Omega)^p}^2 + \sum_{j=1}^p\|\nabla f_j\|_{L^2(\Omega)^p}^2.
\]
%%%%%%%%%%%%%%%%%%%%%%%%%%%%%%%%%%%%%%%%%%%%%%%%%%%%%%%%
Because of the invariance of the Sobolev spaces we have $\mathfrak{H}^1(\Omega)^p ={H}^1(\Omega)^p$ and the norms $\|\cdot\|_{H^1(\Omega)^p}$ and $\|\cdot\|_{\mathfrak{H}^1(\Omega)^p}$ are equivalent (see \cite[Lemma 1.3.]{GR86}). By abuse of notation, both of them will be denoted $\|\cdot\|_{H^1(\Omega)^p}$.

\subsubsection{Sobolev spaces on the boundary}
Let $|s|\leq1$. We recall that $\Sigma$ has no boundary. We define the Sobolev space of $\CC^p$-valued functions $H^s(\Sigma)^p$ as usual (see \cite[\S 2.4.]{SS11}), that is using local coordinates representation on the manifold $\Sigma$.
As $\Sigma$ has no boundary, we have $H^{-s}(\Sigma)^p = (H^s(\Sigma)^p)'$. For $f\in H^{-s}(\Sigma)^p$, 
the norm on $H^{-s}(\Sigma)^d$ can be characterised by duality, that is:
%%%%%%%%%%%%%%%%%%%%%%%%%%%%%%%%%%%%%%%%%%%%%%%%%%%%%%%%
\[
	\|f\|_{H^{-s}(\Sigma)^p} = \sup_{ g\in H^{s}(\Sigma)^p, g\neq0}\frac{\ps{f}{g}_{H^{-s}(\Sigma)^p,H^{s}(\Sigma)^p}}{\|g\|_{H^{s}(\Sigma)^p}}.
\]
%%%%%%%%%%%%%%%%%%%%%%%%%%%%%%%%%%%%%%%%%%%%%%%%%%%%%%%%

For a function $g\in\mathcal{C}_0^\infty(\overline{\Omega})^p$, $\tr g$ denotes its trace on $\Sigma$. $\tr:g\mapsto\tr g$ is a linear operator from $\mathcal{C}_0^\infty(\overline{\Omega})^p$ to $\mathcal{C}^\infty(\Sigma)^p$ and we have the following well-known trace theorem (see, for instance, \cite[Th. 3.37]{Mc00}):
\begin{prop}[Trace theorem] The linear operator $\tr$ extends into a bounded operator from $H^1(\Omega)^p$ to $H^{1/2}(\Sigma)^p$ also denoted $\tr$. Moreover, there exists a bounded linear extension operator $E : H^{1/2}(\Sigma)^p \rightarrow H^1(\Omega)$ satisfying
%%%%%%%%%%%%%%%%%%%%%%%%%%%%%%%%%%%%%%%%%%%%%%%%%%%%%%%%
\[
\tr E(g) = g,\quad \text{for all } g\in H^{1/2}(\Sigma)^p.
\]
%%%%%%%%%%%%%%%%%%%%%%%%%%%%%%%%%%%%%%%%%%%%%%%%%%%%%%%%
\label{prop:tr_th_cla}
\end{prop}

\subsubsection{Dirac operator and fundamental solutions}
Let ${\alpha} = (\alpha_1, \alpha_2, \alpha_3)$ and ${\beta}$ be the $4\times4$ Hermitian and unitary matrices given by:
%%%%%%%%%%%%%%%%%%%%%%%%%%%%%%%%%%%%%%%%%%%%%%%%%%%%%%%%
\[
	\alpha_j=\bigg(\begin{array}{cc}
0 & \sigma_j \\
\sigma_j & 0
\end{array}\bigg)\quad \text{for } j=1,2,3,\quad {\beta} =\bigg(\begin{array}{cc}
I_2 & 0 \\
0 & -I_2
\end{array}\bigg).
\]
%%%%%%%%%%%%%%%%%%%%%%%%%%%%%%%%%%%%%%%%%%%%%%%%%%%%%%%%
Here $(\sigma_1,\sigma_2,\sigma_3)$ are the Pauli matrices defined as
%%%%%%%%%%%%%%%%%%%%%%%%%%%%%%%%%%%%%%%%%%%%%%%%%%%%%%%%
\[
\sigma_1=\bigg(\begin{array}{cc}
0 & 1 \\
1 & 0
\end{array}\bigg),\quad\sigma_2=\bigg(\begin{array}{cc}
0 & -i \\
i & 0
\end{array}\bigg),\quad\sigma_3=\bigg(\begin{array}{cc}
1 & 0 \\
0 & -1
\end{array}\bigg).
\]
%%%%%%%%%%%%%%%%%%%%%%%%%%%%%%%%%%%%%%%%%%%%%%%%%%%%%%%%
The Dirac operator is the differential operator acting on the space of distributions $\mathcal{D}'(\Omega)^4$ defined as
%%%%%%%%%%%%%%%%%%%%%%%%%%%%%%%%%%%%%%%%%%%%%%%%%%%%%%%%
\[
	\mathcal{H}(\mu):= \mathcal{H} = \alpha\cdot\sfD + \mu\beta,\quad \sfD=-i\nabla,
\]
%%%%%%%%%%%%%%%%%%%%%%%%%%%%%%%%%%%%%%%%%%%%%%%%%%%%%%%%
where for $X=(X_1,X_2,X_3)$,  $\alpha\cdot X = \sum_{j=1}^3\alpha_jX_j$.

We introduce the Sobolev space associated with the Dirac operator on the domain $\Omega$ as
%%%%%%%%%%%%%%%%%%%%%%%%%%%%%%%%%%%%%%%%%%%%%%%%%%%%%%%%
\begin{equation}
	H(\alpha,\Omega) := \big\{u \in L^2(\Omega)^4 : \mathcal{H} u\in L^2(\Omega)^4\big\}= \big\{u \in L^2(\Omega)^4 : (\alpha\cdot\sfD) u\in L^2(\Omega)^4\big\},
\end{equation}
%%%%%%%%%%%%%%%%%%%%%%%%%%%%%%%%%%%%%%%%%%%%%%%%%%%%%%%%
where $\mathcal{H} u$ and $(\alpha\cdot\sfD) u$ have to be understood in the sense of distributions. As the multiplication by $\beta$ is a bounded operator from $L^2(\Omega)^4$ onto itself, the equality between these spaces hold and we can endow them with the scalar product
%%%%%%%%%%%%%%%%%%%%%%%%%%%%%%%%%%%%%%%%%%%%%%%%%%%%%%%%
\[
	\ps{u}{v}_{H(\alpha,\Omega)} = \ps{u}{v}_{L^2(\Omega)^4} + \ps{(\alpha\cdot\sfD)u}{(\alpha\cdot\sfD)v}_{L^2(\Omega)^4},\quad u,v\in H(\alpha,\Omega),
\]
%%%%%%%%%%%%%%%%%%%%%%%%%%%%%%%%%%%%%%%%%%%%%%%%%%%%%%%%
it is a Hilbert space (see Section \ref{subsec:Sob_spa} for more details) and the associated norm is denoted $\|\cdot\|_{H(\alpha,\Omega)}$.
\begin{remark}
For any $u\in H(\alpha,\Omega)$, the norm $\|\cdot\|_{H(\alpha,\Omega)}$ is equivalent to the operator norm $\|u\|_{\mathcal{H}} = \|u\|_{L^2(\Omega)^4} + \|\mathcal{H} u\|_{L^2(\Omega)^4}$ and, by abuse of notation, we also denote $\|\cdot\|_{\mathcal{H}}$ by $\|\cdot\|_{H(\alpha,\Omega)}$.
\label{rmk:eq_dir_sob}
\end{remark}

\subsection{Structure of the paper} This paper is organized as follows. In Section \ref{sec:lay_pot_cald} we study the layer potential associated with the Dirac system and introduce various tools that will be helpful in the following, such as the Calder\'on projectors. The main result in this section is Theorem \ref{thm:regu_distri} and its consequences regarding the Calder\'on projectors.

In Section \ref{sec:MIT-bag} we prove Theorem \ref{thm:MIT_sa} about the self-adjointness of the MIT bag model and in Section \ref{sec:Dir-shell} we study the self-adjointness of the Dirac operator coupled with an electrostatic $\delta$-shell interaction, the main result being Theorem \ref{th:sadirac}. Note that the question of self-adjointness for the critical values of the coupling constant that motivated us in the beginning is dealt with.  

\section{Layer potential and Calder\'on projectors for the Dirac System}
\label{sec:lay_pot_cald}
In the following two subsections we state the main results of this section.
\subsection{Trace operator and layer potential}\label{subsec:tr-laypot}
Let $\psi$ denote the following fundamental solution of  $-\Delta^2 +\mu^2$:
%%%%%%%%%%%%%%%%%%%%%%%%%%%%%%%%%%%%%%%%%%%%%%%%%%%%%%%%
\begin{equation}
	\psi(x):= \psi_\mu(x) = \frac{e^{-|\mu x|}}{4\pi|x|},\quad x\in\mathbb{R}^3.
	\label{eqn:fun_lapl_3}
\end{equation}
%%%%%%%%%%%%%%%%%%%%%%%%%%%%%%%%%%%%%%%%%%%%%%%%%%%%%%%%
Since $\mathcal{H}^2 = (-\Delta + \mu^2)\rm{Id}$, $\phi:=\mathcal{H}(\psi\rm{Id})$ is a fundamental solution of $\mathcal{H}$. For $g \in\mathcal{C}^{\infty}(\Sigma)^4$ we define the layer potential
%%%%%%%%%%%%%%%%%%%%%%%%%%%%%%%%%%%%%%%%%%%%%%%%%%%%%%%%%
\begin{equation}
	\Phi(g)(x) := \Phi_{\Omega,\mu}(g)(x) = \int_{y\in\Sigma} \phi(x-y) g(y) \dd \mathfrak{s}(y),\quad x\in\Omega.
\label{eqn:dfn_laypot}
\end{equation}
%%%%%%%%%%%%%%%%%%%%%%%%%%%%%%%%%%%%%%%%%%%%%%%%%%%%%%%%%
%The following proposition is an improvement of \cite[Lemma 2.1.]{AMV14} in a setting adapted to the Dirac operator. It can be deduced from \cite[{\it ii)}~Thm. 1.]{Cos88}.\section{Main results}

We have the following extension of Proposition \ref{prop:tr_th_cla}.
\begin{prop}The trace operator $\mathfrak{t}_\Sigma$ extends into a continuous map $\mathfrak{t}_\Sigma: H(\alpha,\Omega) \rightarrow {H}^{-1/2}(\Sigma)^4$.
\label{prop:transmweak}
\end{prop}

We have the next theorem.
\begin{thm}The following holds:
\begin{itemize}
	\item[i)] If $\mu=0$ and $\Omega$ is unbounded then, for any $R>0$ such that $\Sigma\subset B(R)$, $\Phi$ extends into a bounded operator from $H^{-1/2}(\Sigma)^4$ to $H(\alpha,\Omega\cap B(R))$.
	\item[ii)] Otherwise, $\Phi$ extends into a bounded operator from $H^{-1/2}(\Sigma)^p$ to $H(\alpha,\Omega)$.
\end{itemize}
\label{thm:regu_distri}
\end{thm}

The boundary integral operator is defined taking the boundary data of $\Phi$ on $\Sigma$ (in a distributional sense, see Proposition \ref{prop:tr_th_cla} and Proposition \ref{prop:transmweak}):
%%%%%%%%%%%%%%%%%%%%%%%%%%%%%%%%%%%%%%%%%%%%%%%%%%%%%%%%%
\[
	C_\mu(g):= C(g) = \tr\big(\Phi(g)\big).
\]
%%%%%%%%%%%%%%%%%%%%%%%%%%%%%%%%%%%%%%%%%%%%%%%%%%%%%%%%%
An important consequence of Proposition \ref{prop:transmweak} and Theorem \ref{thm:regu_distri} is that the boundary integral operator $C$ satisfies the following corollary.
\begin{corollary} The following operator is continuous:
%%%%%%%%%%%%%%%%%%%%%%%%%%%%%%%%%%%%%%%%%%%%%%%%%%%%%%%%%
\begin{equation}
	C : H^{-1/2}(\Sigma)^4 \rightarrow H^{-1/2}(\Sigma)^4,
	\label{eqn:def_C}
\end{equation}
%%%%%%%%%%%%%%%%%%%%%%%%%%%%%%%%%%%%%%%%%%%%%%%%%%%%%%%%%
\label{thm:cont_tr}
\end{corollary}

\subsection{Calder\'on projectors}\label{subsec:Cald_proj} The aim of this subsection is to define the Calder\'on projectors and give some of their properties. Set
%%%%%%%%%%%%%%%%%%%%%%%%%%%%%%%%%%%%%%%%%%%%%%%%%%%%%%%%%
\[
	\Omega_+:=\Omega\quad\text{and}\quad\Omega_-:=\mathbb{R}^3\setminus\overline{\Omega}.
\]
%%%%%%%%%%%%%%%%%%%%%%%%%%%%%%%%%%%%%%%%%%%%%%%%%%%%%%%%%
We can define two operators $\Phi_\pm := \Phi_{\Omega_\pm,\mu}$ as in \eqref{eqn:dfn_laypot} which allows us to set $C_{\pm}:= C_{\pm,\mu} =\mathfrak{t}_{\Sigma,\pm} \circ \Phi_\pm$ where $\mathfrak{t}_{\Sigma,\pm}$ denotes the trace operator of Proposition \ref{prop:transmweak} from $H(\alpha,\Omega_\pm)$ to $H^{-1/2}(\Sigma)^4$.

Now, we define the Calder\'on projectors and give some properties.
\begin{dfn} The Calder\'on projectors associated with $\mu\in\mathbb{R}$ are the bounded linear operators from $H^{-1/2}(\Sigma)^4$ onto itself defined as:
%%%%%%%%%%%%%%%%%%%%%%%%%%%%%%%%%%%%%%%%%%%%%%%%%%%%%%%%
\[
	\mathcal{C}_{\pm} := \mathcal{C}_{\pm,\mu} = \pm i C_{\pm}(\alpha\cdot\bfn).
\]
%%%%%%%%%%%%%%%%%%%%%%%%%%%%%%%%%%%%%%%%%%%%%%%%%%%%%%%%
\label{dfn:cald_proj}
\end{dfn}

\begin{remark} As $\Sigma$ is $\mathcal{C}^2$, the multiplication by $\alpha\cdot\bfn$ is a bounded linear operator from $H^{-1/2}(\Sigma)^4$ onto itself. Thus the definition makes sense.
\end{remark}

Before giving the first properties on the Calder\'on projectors we define their formal adjoints as
\[
	\mathcal{C}_\pm^* = \mp i (\alpha\cdot\bfn)C_\mp.
\]
By definition, $\mathcal{C}_\pm^*$ is a linear bounded operator from $H^{-1/2}(\Sigma)^4$ onto itself.
The following proposition justifies that the Calder\'on projectors are actual projectors.
\begin{prop} We have:
%%%%%%%%%%%%%%%%%%%%%%%%%%%%%%%%%%%%%%%%%%%%%%%%%%%%%%%%
\begin{itemize}
	\item[i)] $\mathcal{C}_\pm' = \mathcal{C}_\pm^*|_{H^{1/2}(\Sigma)^4}$ and $(\mathcal{C}_\pm^*)' = \mathcal{C}_\pm|_{H^{1/2}(\Sigma)^4}$. In particular $\mathcal{C}_\pm|_{H^{1/2}(\Sigma)^4}$ and $\mathcal{C}_\pm^*|_{H^{1/2}(\Sigma)^4}$ are bounded operators from $H^{1/2}(\Sigma)^4$ onto itself,
	\item[ii)] $(\mathcal{C}_\pm)^2 = \mathcal{C}_\pm$ and $(\mathcal{C}_\pm^*)^2 = \mathcal{C}_\pm^*$,
	\item[iii)] $\mathcal{C}_+ + \mathcal{C}_- = Id$ and $\mathcal{C}_+^* + \mathcal{C}_-^* = Id$,
	\item[iv)] $(\alpha\cdot \bfn) \mathcal{C}_{\pm} = \mathcal{C}_\mp^*(\alpha\cdot \bfn)$ and $\mathcal{C}_\pm(\alpha\cdot \bfn)=(\alpha\cdot \bfn)\mathcal{C}_\mp^*.$
\end{itemize}
%%%%%%%%%%%%%%%%%%%%%%%%%%%%%%%%%%%%%%%%%%%%%%%%%%%%%%%%
\label{prop:propcaldproj}
\end{prop}

The following two propositions are key-points in the proof of self-adjointness of Dirac operators. Both are regularisation properties related to the Calder\'on projectors.
\begin{prop} The operator $\mathcal{C}_{\pm}\circ\mathfrak{t}_{\Sigma,\mp}$ is a linear bounded operator from $H(\alpha,\Omega_\mp)$ to $H^{1/2}(\Sigma)^p$.
\label{prop:regCaldtr}
\end{prop}

Note that the Calder\'on projectors satisfy:
%%%%%%%%%%%%%%%%%%%%%%%%%%%%%%%%%%%%%%%%%%%%%%%%%%%%%%%%
\begin{equation}
	\mathcal{C}_{\pm} - \mathcal{C}_{\pm}^* = \pm i \mathcal{A},
	\label{eqn:def_anticom}
\end{equation}
%%%%%%%%%%%%%%%%%%%%%%%%%%%%%%%%%%%%%%%%%%%%%%%%%%%%%%%%
where $\mathcal{A} = \mathcal{A}_{\mu}$ does not depend on the sign $\pm$ and is an anticommutator that will be specified in Part \ref{subsubsec:reg_anticom}. Roughly speaking, $\mathcal{A}$ measures the defect of self-adjointness of the Calder\'on projectors.

\begin{prop} The operator $\mathcal{A}$ extends into a bounded operator from $H^{-1/2}(\Sigma)^4$ to $H^{1/2}(\Sigma)^4$.
\label{ref:reg_anti}
\end{prop}

\begin{remark}
Proposition \ref{ref:reg_anti} is reminiscent of \cite[Lemma 3.5]{AMV14} that states that $\mathcal{A}$ is compact as an operator from $L^2(\Sigma)^4$ onto itself. In fact, one can prove that $\mathcal{A}$ is a linear bounded operator from  $L^{2}(\Sigma)^4$ to $H^1(\Sigma)^4$. As the injection $H^1(\Sigma)^4$ into $L^2(\Sigma)^4$ is compact, we recover \cite[Lemma 3.5]{AMV14}.
\label{rmk:compac_anticom}
\end{remark}

The rest of this section is splitted into three subsections. In Subsection \ref{subsec:Sob_spa} we study the Sobolev space $H(\alpha,\Omega)$ in order to prove Proposition \ref{prop:transmweak}. Subsection \ref{subsec:proof_th} deals with Theorem \ref{thm:regu_distri} and Corollary \ref{thm:cont_tr} while the various properties of the Calder\'on projectors are proven in Subsection \ref{subsec:cald-proj}.

\subsection{The Sobolev space $H(\alpha,\Omega)$}
\label{subsec:Sob_spa}
In this subsection, for the sake of clarity, we set $\mathcal{K}:= H(\alpha,\Omega)$. Recall that $\mathcal{K}$ is endowed with the scalar product:
%%%%%%%%%%%%%%%%%%%%%%%%%%%%%%%%%%%%%%%%%%%%%%%%%%%%%%%%
\[
	\ps{u}{v}_\mathcal{K} = \ps{u}{v}_{L^2(\Omega)^4} + \ps{(\alpha\cdot\sfD)u}{(\alpha\cdot\sfD)v}_{L^2(\Omega)^4},\quad u,v\in\mathcal{K}.
\]
%%%%%%%%%%%%%%%%%%%%%%%%%%%%%%%%%%%%%%%%%%%%%%%%%%%%%%%%
Let $\|\cdot\|_\mathcal{K}$ denotes the norm associated with the scalar product $\ps{\cdot}{\cdot}_{\mathcal{K}}$.

We aim to prove Proposition \ref{prop:transmweak} in order to give a sense to the boundary value of a function in $\mathcal{K}$.

This subsection is organised as follows: In Part \ref{subsec:basic_pt} we give basic properties of the space $\mathcal{K}$ and Proposition \ref{prop:transmweak} is proven in Part \ref{subsec:tr_th_K}.

\subsubsection{Basic properties}
\label{subsec:basic_pt}
Let us start with the following properties.
\begin{prop} $(\mathcal{K},\|\cdot\|_\mathcal{K})$ is a Hilbert space.
\label{prop:K_hilbsp}
\end{prop}
\begin{proof}[Proof of Proposition \ref{prop:K_hilbsp}] Let $u_n$ be a Cauchy sequence of $(\mathcal{K},\|\cdot\|_\mathcal{K})$. In particular $u_n$ and $(\alpha\cdot\sfD) u_n$ are Cauchy sequences of $L^2(\Omega)^4$ and converge:
%%%%%%%%%%%%%%%%%%%%%%%%%%%%%%%%%%%%%%%%%%%%%%%%%%%%%%%%
\[
	u_n \underset{n\rightarrow+\infty}{\longrightarrow} u \in L^2(\Omega)^4,\quad (\alpha\cdot\sfD)u_{n} \underset{n\rightarrow+\infty}{\longrightarrow} v\in L^2(\Omega)^4.
\]
%%%%%%%%%%%%%%%%%%%%%%%%%%%%%%%%%%%%%%%%%%%%%%%%%%%%%%%%
In the sense of distributions we have
%%%%%%%%%%%%%%%%%%%%%%%%%%%%%%%%%%%%%%%%%%%%%%%%%%%%%%%%
\[
	(\alpha\cdot\sfD) u = \lim_{n\rightarrow\infty} (\alpha\cdot\sfD) u_{n} = v.
\]
%%%%%%%%%%%%%%%%%%%%%%%%%%%%%%%%%%%%%%%%%%%%%%%%%%%%%%%%
Consequently $(\alpha\cdot\sfD) u = v$ in $L^2(\Omega)^4$ and $u_n$ converges to $u$ in $\mathcal{K}$.
\end{proof}
\begin{prop} The inclusion of $H^1(\Omega)^4$ into $\mathcal{K}$ is continuous. More precisely, there exists $c>0$ such that for all $u\in H^1(\Omega)^4$:
%%%%%%%%%%%%%%%%%%%%%%%%%%%%%%%%%%%%%%%%%%%%%%%%%%%%%%%%
\[
	\|u\|_{\mathcal{K}}^2 \leq c \|u\|_{H^1(\Omega)^4}^2.
\]
%%%%%%%%%%%%%%%%%%%%%%%%%%%%%%%%%%%%%%%%%%%%%%%%%%%%%%%%
\label{prop:prop_incl_ctn}
\end{prop}

\begin{proof}[Proof of Proposition \ref{prop:prop_incl_ctn}]Let $u=(u_k)_{k\in\{1,\dots,4\}} \in H^1(\Omega)^4$, we have:
%%%%%%%%%%%%%%%%%%%%%%%%%%%%%%%%%%%%%%%%%%%%%%%%%%%%%%%%
\begin{equation}
	\begin{array}{lcl}
		\|(\alpha\cdot\sfD) u\|_{L^2(\Omega)^4} &\leq &\displaystyle\sum_{j=1}^3\|\alpha_j\partial_j u\|_{L^2(\Omega)^4}\\
		&\leq& \displaystyle\sum_{j=1}^3\|\alpha_j\|_{\mathcal{M}}\|\partial_j u\|_{L^2(\Omega)^4} = \sum_{j=1}^3\|\partial_j u\|_{L^2(\Omega)^4}.\\
	\end{array}
	\label{eqn:inf_born}
\end{equation}
%%%%%%%%%%%%%%%%%%%%%%%%%%%%%%%%%%%%%%%%%%%%%%%%%%%%%%%%
Consequently, there exists $c>0$ such that $\displaystyle \|(\alpha\cdot\sfD)u\|_{L^2(\Omega)^4}^2 \leq c'\displaystyle\sum_{k=1}^4\|\nabla u_k\|_{L^2(\Omega)^4}^2$. We obtain Proposition \ref{prop:prop_incl_ctn} with $c = 1 +c'$.
\end{proof}
Now, we state a density result.
\begin{prop}$\mathcal{C}_0^\infty(\overline{\Omega})^4$ is dense in $\mathcal{K}$ for the norm $\|\cdot\|_\mathcal{K}$.
\label{prop:dens_K}
\end{prop}
Before going through the proof of Proposition \ref{prop:dens_K} we state the following lemma. Its proof is a simple consequence of Green's formula and is omitted.
\begin{lem} The following set inclusion holds $\big\{u\in L^2(\mathbb{R}^3)^4 : (\alpha\cdot\sfD) u\in L^2(\mathbb{R}^3)^4 \big\} = H^1(\mathbb{R}^3)^4$.
\label{ref:lemdis}
\end{lem}
Now we have all the tools to prove Proposition \ref{prop:dens_K}.
\begin{proof}[Proof of Proposition \ref{prop:dens_K}]
Let $v\in\mathcal{K}$ such that, for all $u\in\mathcal{C}_0^\infty(\overline{\Omega})^4$ we have:
%%%%%%%%%%%%%%%%%%%%%%%%%%%%%%%%%%%%%%%%%%%%%%%%%%%%%%%%
\[
	0 = \ps{v}{u}_\mathcal{K} = \ps{v}{u}_{L^2(\Omega)^4} + \ps{(\alpha\cdot\sfD)v}{(\alpha\cdot\sfD)u}_{L^2(\Omega)^4}.
\]
%%%%%%%%%%%%%%%%%%%%%%%%%%%%%%%%%%%%%%%%%%%%%%%%%%%%%%%%
Let $w:=(\alpha\cdot\sfD)v$. In $\mathcal{D}'(\Omega)^4$ we have $(\alpha\cdot\sfD)w = - v$ and then the equality is also true in $L^2(\Omega)^4$. Let $w_0$ and $v_0$ be the extensions of $w$ and $v$ by zero to $\mathbb{R}^3$. For any $f\in\mathcal{C}_0^\infty(\mathbb{R}^3)^4$, we have
%%%%%%%%%%%%%%%%%%%%%%%%%%%%%%%%%%%%%%%%%%%%%%%%%%%%%%%%
\[
	\ps{(\alpha\cdot\sfD)w_0}{f}_{\mathcal{D}'(\mathbb{R}^3)^4,\mathcal{D}(\mathbb{R}^3)^4} = \ps{w_0}{(\alpha\cdot\sfD)f}_{L^2(\mathbb{R}^3)^4} = \ps{w}{(\alpha\cdot\sfD)f}_{L^2(\Omega)^4} = -\ps{v}{f}_{L^2(\Omega)^4} = -\ps{v_0}{f}_{L^2(\mathbb{R}^3)^4}.
\]
%%%%%%%%%%%%%%%%%%%%%%%%%%%%%%%%%%%%%%%%%%%%%%%%%%%%%%%%
Thus, $(\alpha\cdot\sfD)w_0 = -v_0 \in L^2(\mathbb{R}^3)^4$. By Lemma \ref{ref:lemdis}, $w_0\in H^1(\mathbb{R}^3)^4$ and finally $w\in H_0^1(\Omega)^4$ thanks to \cite[Prop IX.18]{Br94}.

Let $(f_n)_{n\in\mathbb{N}}$ be a sequence of $\mathcal{C}_0^\infty(\Omega)^4$-functions such that $f_n$ converges to $w$ in the $\|\cdot\|_{H^1(\Omega)^4}$-norm. We have:
%%%%%%%%%%%%%%%%%%%%%%%%%%%%%%%%%%%%%%%%%%%%%%%%%%%%%%%%
\[
	\begin{array}{lcl}
	0 = \ps{v}{u}_{L^2(\Omega)^4} + \ps{(\alpha\cdot\sfD)v}{(\alpha\cdot\sfD)u}_{L^2(\Omega)^4} &=& \ps{v}{u}_{L^2(\Omega)^4} + \ps{w}{(\alpha\cdot\sfD)u}_{L^2(\Omega)^4}\\
	&=&\displaystyle\ps{v}{u}_{L^2(\Omega)^4} + \lim_{n\rightarrow+\infty}\ps{(\alpha\cdot\sfD)f_n}{u}_{L^2(\Omega)^4}\\
	&=&\displaystyle\ps{v}{u}_{L^2(\Omega)^4} + \ps{(\alpha\cdot\sfD)w}{u}_{L^2(\Omega)^4}.
	\end{array}
\]
%%%%%%%%%%%%%%%%%%%%%%%%%%%%%%%%%%%%%%%%%%%%%%%%%%%%%%%%
In particular, $v + (\alpha\cdot\sfD) w = 0$ in $L^2(\Omega)^4$ which gives $w - \Delta w = 0$ in $\mathcal{D}'(\Omega)^4$. We get:
\begin{align*}
	\ps{w}{f_n}_{L^2(\Omega)^4} = \ps{\Delta w}{f_n}_{\mathcal{D}'(\Omega)^4,\mathcal{D}(\Omega)^4} &= -\ps{(\alpha\cdot\sfD)w}{(\alpha\cdot\sfD)f_n}_{\mathcal{D}'(\Omega)^4,\mathcal{D}(\Omega)^4} \\&= -\ps{(\alpha\cdot\sfD)w}{(\alpha\cdot\sfD)f_n}_{L^2(\Omega)^4}.
\end{align*}
Now, letting $n\rightarrow +\infty$, we obtain
\[
	\|w\|_{L^2(\Omega)^4}^2 = - \|(\alpha\cdot\sfD)w\|_{L^2(\Omega)^4}^2,
\]
thus $w = 0$ and $w_0 = 0$. Now, recall that $v_0 = -(\alpha\cdot\sfD)w_0 =0$ which gives $v=0$.
\end{proof}

\subsubsection{Trace theorem}
\label{subsec:tr_th_K}
In this subsection we prove Proposition \ref{prop:transmweak}. To do so, we need the following lemma which is a basic application of Green's formula and whose proof will be therefore omitted.

\begin{lem}
Let $u,v\in\mathcal{C}_0^\infty(\overline{\Omega})^4$, we have:
%%%%%%%%%%%%%%%%%%%%%%%%%%%%%%%%%%%%%%%%%%%%%%%%%%%%%%%%%%%%%
\[
	\ps{(\alpha\cdot\sfD)u}{v}_{L^2(\Omega)^4} = \ps{u}{(\alpha\cdot\sfD)v}_{L^2(\Omega)^4} + \ps{(-i\alpha\cdot\bfn)\tr u}{\tr v}_{L^2(\Sigma)^4}. 
\]
%%%%%%%%%%%%%%%%%%%%%%%%%%%%%%%%%%%%%%%%%%%%%%%%%%%%%%%%%%%%%
By density of $\mathcal{C}_0^\infty(\overline{\Omega})^4$ and continuity of $\tr$ on $H^1(\Omega)^4$, this equality extends to $u,v\in H^1({\Omega})^4$.
\label{lem:Green_dir}
\end{lem}

Now have all the tools to prove Proposition \ref{prop:transmweak}.

\begin{proof}[Proof of Proposition \ref{prop:transmweak}]Let $v\in\mathcal{K}$, we prove that the trace $\tr v$ exists and is in $H^{-1/2}(\Sigma)^4$. Let $(v_{n})_{n\in\mathbb{N}}$ be a sequence of $\mathcal{C}_0^\infty(\overline{\Omega})^4$ converging to $v$ in the $\|\cdot\|_{\mathcal{K}}$-norm.

We want to prove that $\tr v_{n}$ converges in $H^{-1/2}(\Sigma)^4$ and to do so, we prove that it is a Cauchy sequence. For all $f\in H^{1/2}(\Sigma)^4$, Lemma \ref{lem:Green_dir} yields:
%%%%%%%%%%%%%%%%%%%%%%%%%%%%%%%%%%%%%%%%%%%%%%%%%%%%%%%%%%%%%
\begin{equation}
	\ps{(-i\alpha\cdot\bfn)f}{\tr v_n}_{L^2(\Sigma)^4} = \ps{(\alpha\cdot\sfD) \big(E(f)\big)}{v_{n}}_{L^2(\Omega)^4} - \ps{E(f)}{(\alpha\cdot\sfD)v_{n}}_{L^2(\Omega)^4},
	\label{eqn:green_weaktrace}
\end{equation}
%%%%%%%%%%%%%%%%%%%%%%%%%%%%%%%%%%%%%%%%%%%%%%%%%%%%%%%%%%%%%
where $E$ is the extension operators defined in Proposition \ref{prop:tr_th_cla}.

Using the Cauchy-Schwarz inequality and Proposition \ref{prop:prop_incl_ctn} we get
%%%%%%%%%%%%%%%%%%%%%%%%%%%%%%%%%%%%%%%%%%%%%%%%%%%%%%%%%%%%%
\[
|\ps{(\alpha\cdot\bfn) f}{\tr(v_{n} - v_{m})}_{L^2(\Sigma)^4}| \leq  \|E(f))\|_{H^1(\Omega)^4}\big(\|v_{n} - v_{m}\|_{L^2(\Omega)^4} + \|(\alpha\cdot\sfD)(v_{n} - v_{m})\|_{L^2(\Omega)^4}\big).
\]
%%%%%%%%%%%%%%%%%%%%%%%%%%%%%%%%%%%%%%%%%%%%%%%%%%%%%%%%%%%%%
Proposition \ref{prop:tr_th_cla} yields the existence of a constant $c>0$ such that
%%%%%%%%%%%%%%%%%%%%%%%%%%%%%%%%%%%%%%%%%%%%%%%%%%%%%%%%%%%%%
\[
|\ps{(\alpha\cdot\bfn) f}{\tr(v_{n} - v_{m})}_{L^2(\Sigma)^4}| \leq  c\|f\|_{H^{1/2}(\Sigma)^4}\|v_n - v_m\|_\mathcal{K}.
\]
%%%%%%%%%%%%%%%%%%%%%%%%%%%%%%%%%%%%%%%%%%%%%%%%%%%%%%%%%%%%%
As a multiplication operator from $L^2(\Sigma)^4$ onto itself, $\alpha\cdot\bfn$ is self-adjoint and we get:
%%%%%%%%%%%%%%%%%%%%%%%%%%%%%%%%%%%%%%%%%%%%%%%%%%%%%%%%%%%%%
\[
\|(\alpha\cdot\bfn) \tr(v_{n} - v_{m})\|_{H^{-1/2}(\Sigma)^4} = \sup_{f\in H^{1/2}(\Sigma)^4, f\neq0}\frac{\big|\ps{f}{(\alpha\cdot\bfn)\tr(v_{n} - v_{m})}_{L^2(\Sigma)^4}\big|}{\|f\|_{H^{1/2}(\Sigma)^4}} \leq  c\|v_n - v_m\|_\mathcal{K}.
\]
%%%%%%%%%%%%%%%%%%%%%%%%%%%%%%%%%%%%%%%%%%%%%%%%%%%%%%%%%%%%%
As $(v_{n})_{n\in\mathbb{N}}$ converges in the $\|\cdot\|_\mathcal{K}$-norm, $\big((\alpha\cdot\bfn)\tr v_{n}\big)_{n\in\mathbb{N}}$ is a Cauchy-sequence and converges to an element in $H^{-1/2}(\Sigma)^4$. Now, remark that for all $x\in\Sigma$, $(\alpha\cdot\bfn(x))^2 = Id$ and, as $\Sigma$ is $\mathcal{C}^2$-smooth, $\alpha\cdot \bfn (x)$ has $\mathcal{C}^1$-coefficients. Thus, the multiplication by $\alpha\cdot\bfn$ extends into a linear bounded operator from $H^{1/2}(\Sigma)^4$ onto itself and $(\tr v_n)_{n\in\mathbb{N}}$ is  Cauchy sequence in $H^{-1/2}(\Sigma)^4$.

Starting from \eqref{eqn:green_weaktrace} with $v_{n}$ instead of $v_{n} - v_{m}$ and reproducing the same argument we get:

%%%%%%%%%%%%%%%%%%%%%%%%%%%%%%%%%%%%%%%%%%%%%%%%%%%%%%%%%%%%%
\[
\|\tr v_{n}\|_{H^{-1/2}(\Sigma)^4} = \sup_{f\in H^{1/2}(\Sigma)^4, f\neq0}\frac{\big|\ps{f}{\tr v_{n}}_{L^2(\Sigma)^4}\big|}{\|f\|_{H^{1/2}(\Sigma)^4}} \leq  c\|v_n\|_\mathcal{K}.
\]
%%%%%%%%%%%%%%%%%%%%%%%%%%%%%%%%%%%%%%%%%%%%%%%%%%%%%%%%%%%%%
Letting $n\rightarrow+\infty$ we finally obtain the continuity of the trace operator:
%%%%%%%%%%%%%%%%%%%%%%%%%%%%%%%%%%%%%%%%%%%%%%%%%%%%%%%%%%%%%
\[
\|\tr v\|_{H^{-1/2}(\Sigma)^4} = \sup_{f\in H^{1/2}(\Sigma)^4, f\neq0}\frac{\big|\ps{f}{\tr v_{n}}_{L^2(\Sigma)^4}\big|}{\|f\|_{H^{1/2}(\Sigma)^4}} \leq  c\|v\|_\mathcal{K}.
\]
%%%%%%%%%%%%%%%%%%%%%%%%%%%%%%%%%%%%%%%%%%%%%%%%%%%%%%%%%%%%%
\end{proof}
%
%\begin{remark} In the proof of Proposition \ref{prop:transmweak}, the $\mathcal{C}^2$-regularity of $\Sigma$ ensures that the multiplication by $\alpha\cdot \bfn$ is a linear bounded operator from $H^{-1/2}(\Sigma)^p$ onto itself.
%\end{remark}
As a direct corollary, we can extend the Green's formula as follows.
\begin{cor}Let $u\in H(\alpha,\Omega)$ and $v\in H^1(\Omega)^4$, we have:
%%%%%%%%%%%%%%%%%%%%%%%%%%%%%%%%%%%%%%%%%%%%%%%%%%%%%%%%
\[
	\ps{(\alpha\cdot\sfD) u}{ v}_{L^2(\Omega)^4} - \ps{u}{(\alpha\cdot\sfD) v}_{L^2(\Omega)^4} = \ps{(-i\alpha\cdot\bfn)\tr u}{\tr v}_{H^{-1/2}(\Sigma)^4,H^{1/2}(\Sigma)^4}.
\]
%%%%%%%%%%%%%%%%%%%%%%%%%%%%%%%%%%%%%%%%%%%%%%%%%%%%%%%%
\label{prop:Green_weak}
\end{cor}

\subsubsection{Regularisation {\it via} traces}
\label{subsec:regu_tra}
In this part we prove that if the trace of a function in the Sobolev space $H(\alpha,\Omega)$ is sufficiently regular, then $u$ belongs to the usual space $H^1(\Omega)^4$.
\begin{prop} Let $u\in H(\alpha,\Omega)$. Assume that $\tr u\in H^{1/2}(\Sigma)^4$, then we have $u\in H^{1}(\Omega)^4$.
\label{prop:regul_trac}
\end{prop}
\begin{proof}[Proof of Proposition \ref{prop:regul_trac}]
Let $u\in H(\alpha,\Omega)$ be such that $\tr uÊ\in H^{1/2}(\Sigma)^4$. Let us replace $u$ by $u - E(\tr u)$, with $E$ the extension operator of Proposition \ref{prop:tr_th_cla}. Hence, we can assume that $\tr u =0$. Let $u_0$ and $w_0$ be the extension of $u$ and $\mathcal{H} u $ by zero to the whole space $\mathbb{R}^3$, respectively. We have:
\[
	\ps{\mathcal{H} u_0}{v}_{\mathcal{D}'(\mathbb{R}^3)^4,\mathcal{D}(\mathbb{R}^3)^4} = \ps{u_0}{\mathcal{H} v}_{\mathcal{D}'(\mathbb{R}^3)^4,\mathcal{D}(\mathbb{R}^3)^4} = \ps{u}{\mathcal{H} v}_{L^2(\Omega)^4}.
\]
Thanks to Corollary \ref{prop:Green_weak}, we get
\[
\ps{u}{\mathcal{H} v}_{L^2(\Omega)^4} = \ps{\mathcal{H} u}{v}_{L^2(\Omega)^4} = \ps{w_0}{v}_{\mathcal{D}'(\mathbb{R}^3)^4,\mathcal{D}(\mathbb{R}^3)^4}.
\]
Thus, we obtain the following equality in $\mathcal{D}'(\mathbb{R}^3)^4$:
\[
\mathcal{H} u_0 =w_0.
\]
The right-hand side is in $L^2(\mathbb{R}^3)^4$ and thus $u_0\in H^1(\mathbb{R}^3)^4$ by Lemma \ref{ref:lemdis}. We end up with $u\in H^1(\Omega)^4$.
\end{proof}

\subsection{Boundary integral operators}
\label{subsec:proof_th}

In this subsection we aim to prove Theorem \ref{thm:regu_distri}. We follow the usual strategy to prove regularity properties of the usual single and double layer potential (see, for instance, \cite{Cos88} and the book \cite[Chapter 3]{SS11}).

First, in Part \ref{subsec:bound_int} we study an operator which is reminescent of the Newtonian potential. Second, in Part \ref{subsec:proof_th1}, we give a new definition of the layer potential that extends \eqref{eqn:dfn_laypot} and prove Theorem \ref{thm:regu_distri}.

\subsubsection{Layer potential of the Dirac operator}
\label{subsec:bound_int}
Let $R>0$ be such that $\Sigma\subset B(R)$. We introduce the open domain $\Omegat $ of $\mathbb{R}^3$ as
%%%%%%%%%%%%%%%%%%%%%%%%%%%%%%%%%%%%%%%%%%%%%%%%%%%%%%%%%%%%%%%%%%
\[
	\Omegat := 	\left\{
					\begin{array}{ll}
					\Omega\cap B(R)&\text{if } \mu=0,\\
					\Omega & \text{otherwise}
					\end{array}
				\right..
\]
%%%%%%%%%%%%%%%%%%%%%%%%%%%%%%%%%%%%%%%%%%%%%%%%%%%%%%%%%%%%%%%%%%
All along this section the Sobolev space $H(\alpha,\Omegat)$ will be denoted $\mathcal{K}$. According to Remark \ref{rmk:eq_dir_sob}, for the sake of simplicity, we choose in this section the operator norm
%%%%%%%%%%%%%%%%%%%%%%%%%%%%%%%%%%%%%%%%%%%%%%%%%%%%%%%%%%%%%%%%%%
\[
	\|u\|_{\mathcal{K}}^2 = \|u\|_{L^2(\Omegat)^4}^2 + \|\mathcal{H}u\|_{L^2(\Omegat)^4}^2,\quad \text{ for all } uÊ\in \mathcal{K}.
\]
%%%%%%%%%%%%%%%%%%%%%%%%%%%%%%%%%%%%%%%%%%%%%%%%%%%%%%%%%%%%%%%%%%
Let $f\in\mathcal{S}\big(\mathbb{R}^3)^4$, we define
%%%%%%%%%%%%%%%%%%%%%%%%%%%%%%%%%%%%%%%%%%%%%%%%%%%%%%%%%%%%%%%%%%
\[
		(\widehat{V} f)(x) := \int_{\mathbb{R}^3} \phi(x-y)f(y)\dd y,\quad x \in \mathbb{R}^3.
\]
%%%%%%%%%%%%%%%%%%%%%%%%%%%%%%%%%%%%%%%%%%%%%%%%%%%%%%%%%%%%%%%%%%
By definition $\widehat V f \in \mathcal{S}(\mathbb{R}^3)^4$. For $f,g\in\mathcal{S}(\mathbb{R}^3)^4$ we have
%%%%%%%%%%%%%%%%%%%%%%%%%%%%%%%%%%%%%%%%%%%%%%%%%%%%%%%%%%%%%%%%%%
\begin{equation}
	\ps{\widehat{V} f}{g}_{L^2(\mathbb{R}^3)^4} = \int_{\mathbb{R}^3}\int_{\mathbb{R}^3}\ps{f(x)}{{\phi(y-x)g(y)}}_{\mathbb{C}^4}\dd y\dd x = \ps{f}{\widehat{V} g}_{L^2(\mathbb{R}^3)^4},
	\label{eqn:sym_V}
\end{equation}
%%%%%%%%%%%%%%%%%%%%%%%%%%%%%%%%%%%%%%%%%%%%%%%%%%%%%%%%%%%%%%%%%%
hence we can define $\widehat{V} : \mathcal{S}'(\mathbb{R}^3)^4 \rightarrow \mathcal{S}'(\mathbb{R}^3)^4$. Now, if $f\in L^2(\Omegat)^4$, $f_0$ denotes its extension by $0$ to $L^2(\mathbb{R}^3)^4$. We have $f_0 \in \mathcal{S}'(\mathbb{R}^3)^4$ and we can define the potential $\widehat{V}(f_0)\in\mathcal{S}'(\mathbb{R}^3)^4$.

The aim of this subsection is to prove the following proposition.
\begin{prop}
	The operator
%%%%%%%%%%%%%%%%%%%%%%%%%%%%%%%%%%%%%%%%%%%%%%%%%%%%%%%%%%%%%%%%%%
\[
	V :\left\{	\begin{array}{lcl}
							L^{2}(\Omegat)^4 	& \rightarrow 	& H^{1}(\Omegat)^4\\
							f				& \mapsto		& (\widehat{V}f_0)|_{\Omegat}
						\end{array}\right.
\]
%%%%%%%%%%%%%%%%%%%%%%%%%%%%%%%%%%%%%%%%%%%%%%%%%%%%%%%%%%%%%%%%%%
defines a bounded linear operator.
	\label{prop:pot_regul}
\end{prop}

Note that by symmetry and Proposition \ref{prop:pot_regul}, we know that $V$ is a bounded linear operator from $\widetilde{H}^{-1}(\Omegat)^4 := \big(H^1(\Omegat)^4\big)'$ to $L^2(\Omegat)^4$.

\begin{proof}[Proof of Proposition \ref{prop:pot_regul}] Let $f\in L^2(\Omegat)^4$ and $f_0$ be its extension by zero to the whole $\mathbb{R}^3$. Let us assume that $\mu\neq0$. As $\phi$ is a fundamental solution of $\mathcal{H}(\mu)$, for $\xi\in\mathbb{R}^3$ we have
\[
	\mathcal{F}(\phi)(\xi) = \frac{\alpha\cdot\xi + \mu}{|\xi|^2 + \mu^2}.
\]
Thus, for all $\xi\in\mathbb{R}^3$, we get
\[
	\mathcal{F}(\hat{V}f_0)(\xi) = \mathcal{F}(\phi * f_0)(\xi) = \mathcal{F}(\phi)(\xi)\mathcal{F}(f_0)(\xi) = \frac{\alpha\cdot\xi + \mu}{|\xi|^2 + \mu^2} \mathcal{F}(f_0)(\xi).
\]
Hence, there exists $C>0$ such that
\[
	(1 + |\xi|^2)^{1/2}\|\mathcal{F}(\hat{V}f_0)(\xi)\|_{\mathbb{C}^4} \leq C \|\mathcal{F}(f_0)(\xi)\|_{\mathbb{C}^4},
\]
which yields
\[
	\|\hat{V}f_0\|_{H^1(\mathbb{R}^3)^4} \leq C \|f\|_{L^2(\Omegat)^4}.
\]
By definition, $\|Vf\|_{H^1(\Omegat)^4}\leq\|\hat{V}f_0\|_{H^1(\mathbb{R}^3)^4}$ and the proposition is proven.

If $\mu=0$, following the same strategy, one can prove that there exists $C>0$ such that
\[
	\|\nabla (Vf)\|_{L^2(\Omegat)^4} \leq \|\nabla (\hat{V}f_0)\|_{L^2(\mathbb{R}^3)^4} = \bigg(\int|\xi|^2\|\mathcal{F}(\hat{V}f_0)\|_{\mathbb{C}^4}^2\dd\xi\bigg)^{1/2}\leq C \|f\|_{L^2(\Omegat)^4}.
\]
Now, let $\chi$ be a $\mathcal{C}_0^\infty$-smooth cut-off function non-negative and non-increasing such that $\chi(r) = 1$ for $r\in[0,2R]$ and $\chi(r)=0$ if $r>3R$. Define
%%%%%%%%%%%%%%%%%%%%%%%%%%%%%%%%%%%%%%%%%%%%%%%%%%%%%%%%%%%%%%%%%%
\begin{equation}
	u_{\chi} (x) := \int_{\Omegat}\chi(|x-y|)\phi(x-y)f(y)\dd y,\quad \text{for } x\in\mathbb{R}^3.
	\label{eqn:defmu_m}
\end{equation}
%%%%%%%%%%%%%%%%%%%%%%%%%%%%%%%%%%%%%%%%%%%%%%%%%%%%%%%%%%%%%%%%%%
As $\chi(|x-y|) = 1$ for $x,y\in\Omegat$, we get $u_{\chi}(x) = (Vf) (x)$ for all $x\in\Omegat$. Moreover, we have
\[
	\|Vf\|_{L^2(\Omegat)^4} = \|u_\chi\|_{L^2(\Omegat)^4} \leq \|u_\chi\|_{L^2(\mathbb{R}^3)^4} =\|(\chi\phi)*f_0\|_{L^2(\mathbb{R}^3)^4} \leq \|\chi\Phi\|_{L^1(\mathbb{R}^3;\ \mathbb{C}^{4\times4})}\|f\|_{L^2(\Omegat)^4},
\]
where we used Young's inequality because, thanks to the cut-off, $\chi\phi$ is in the space of integrable functions with values in $\mathbb{C}^{4\times4}$. It concludes the proof.
\end{proof}

\subsubsection{Proof of Theorem \ref{thm:regu_distri}}
\label{subsec:proof_th1}
In this part, we finally prove Theorem \ref{thm:regu_distri}. To do so, we give a new definition of the layer operator \eqref{eqn:dfn_laypot}.
\begin{dfn} The layer potential $\mathcal{L}$ of the Dirac operator is defined as
\[\mathcal{L} = V \circ \tr',
\]
where $\tr$ is the trace operator defined in Proposition \ref{prop:tr_th_cla}.
\label{prop:comp_Phi}
\end{dfn}

\begin{remark} By definition, $\mathcal{L}$ is a linear bounded operator from $H^{-1/2}(\Sigma)^4$ to $L^2(\Omegat)^4$.
\end{remark}
The following proposition states that $\mathcal{L}$ is an extension of $\Phi$.
\begin{prop} Let $g\in\mathcal{C}^\infty(\Sigma)^4$, we have $\mathcal{L}g = \Phi(g)$ in $\Omegat$.
\label{prop:coin_def}
\end{prop}
The proof of Proposition \ref{prop:coin_def} is inspired of what is usually done for the classical single and double layer potentials (see, for instance, \cite[Thm. 3.1.6 b)]{SS11}) 
\begin{proof}[Proof of Proposition \ref{prop:coin_def}] Let $g\in\mathcal{C}^\infty(\Sigma)^4$ and $u\in\mathcal{C}_0^\infty(\Omegat)^4$. We have:
\[
	\ps{\mathcal{L} g}{u}_{L^2(\Omegat)^4} = \ps{g}{\tr V u}_{H^{-1/2}(\Sigma)^4,H^{1/2}(\Sigma)^4} = \ps{g}{\tr V u}_{L^2(\Sigma)^4},
\]
where the last equality holds because $g$ is smooth. A computation yields:
\[
	\begin{split}
	\ps{g}{\tr V u}_{L^2(\Sigma)^4} & = \int_{x\in\Sigma} \ps{g(x)}{\int_{y\in\mathbb{R}^3}\phi(x-y)u(y)\dd y}_{\mathbb{C}^4}\dd\mathfrak{s}(x)\\
	& = \int_{y\in\mathbb{R}^3}\int_{x\in\Sigma}\ps{\phi(y-x)g(x)}{u(y)}\dd\mathfrak{s}(x)\dd y\\
	& = \ps{\Phi(g)}{u}_{L^2(\Omegat)^4},
	\end{split}
\]
where we used that $\phi(x-y)^* = \phi(y-x)$ as well as Fubini's theorem.

Now, take $x\in\Omegat$ and $\mathcal{U}$ a compact neighbourhood of $x$. One can check that $\Phi(g)|_\mathcal{U}\in\mathcal{C}^\infty(\mathcal{U})^4$. The restriction of functions $\mathcal{C}_0^\infty(\Omegat)$ to $\mathcal{U}$ is dense in $L^2(\mathcal{U})$ and consequently we get $\Phi(g) = \mathcal{L} g$ on $\mathcal{U}$ in $L^2(\mathcal{U})^4$.

\end{proof}

From now on, we drop the notation $\mathcal{L}$ and keep denoting $\Phi$ the layer potential. Theorem \ref{thm:regu_distri} is an immediate consequence of the following proposition.

\begin{prop} For all $g\in H^{-1/2}(\Sigma)^4$, we have $\mathcal{H}(\Phi(g)) = 0$ in $\mathcal{D}'(\Omegat)^4$. In particular, $\Phi$ is a linear and bounded operator from $H^{-1/2}(\Sigma)^4$ to $H(\alpha,\Omegat)$.
\label{prop:nul_distrib}
\end{prop}

\begin{proof}[Proof of Proposition \ref{prop:nul_distrib}] Let $g\in H^{-1/2}(\Sigma)^4$ and $u\in\mathcal{C}_0^\infty(\Omegat)$. We have:
\[
	\begin{split}
	\ps{\mathcal{H}(\Phi(g))}{u}_{\mathcal{D}'(\Omegat)^4,\mathcal{D}(\Omegat)^4} = \ps{\Phi(g)}{\mathcal{H}u}_{\mathcal{D}'(\Omegat)^4,\mathcal{D}(\Omegat)^4} & = \ps{\Phi(g)}{\mathcal{H}u}_{L^2(\Omegat)^4}\\ &= \ps{g}{(\tr\circ V) (\mathcal{H}u)}_{H^{-1/2}(\Sigma)^4,H^{1/2}(\Sigma)^4}.
\end{split}
\]
Now, it is an exercise to check that if $u\in\mathcal{C}_0^\infty(\Omegat)$ then $V\mathcal{H} u = u$. Hence, we obtain
\[
	\ps{\mathcal{H}(\Phi(g))}{u}_{\mathcal{D}'(\Omegat)^4,\mathcal{D}(\Omegat)^4} = \ps{g}{\tr u}_{H^{-1/2}(\Sigma)^4,H^{1/2}(\Sigma)^4} = 0,
\]
where the last equality holds because $u$ has compact support in $\Omegat$.
\end{proof}
Now, we can prove Corollary \ref{thm:cont_tr}.
\begin{proof}[Proof of Corollary \ref{thm:cont_tr}]
By definition, $C = \mathfrak{t}_{\Sigma}\circ\Phi$. By composition, thanks to Proposition \ref{prop:transmweak} and Theorem \ref{thm:regu_distri}, we get that $C$ is a linear bounded operator from $H^{-1/2}(\Sigma)^4$ onto itself.
\end{proof}

\subsection{Properties of the Calder\'on projectors}\label{subsec:cald-proj}
In this subsection we prove the propositions stated in Subsection \ref{subsec:Cald_proj} about the Calder\'on projectors. In Part \ref{subsec:alg_rel} we prove Proposition \ref{prop:propcaldproj} and Proposition \ref{prop:regCaldtr}. Proposition \ref{ref:reg_anti} is proven in Part \ref{subsubsec:reg_anticom}, respectively.
\subsubsection{Algebraic relations}
\label{subsec:alg_rel}
The aim of this subsection is to prove Proposition \ref{prop:propcaldproj}. Before going any further, we need to introduce some notations. For a function $f\in\mathcal{C}^\infty(\Sigma)^4$, we introduce the boundary singular integral operator
%%%%%%%%%%%%%%%%%%%%%%%%%%%%%%%%%%%%%%%%%%%%%%%%%%%%%%%%
\begin{equation}
	C_\mathfrak{s}(g)(x) := \lim_{\epsilon \rightarrow 0} \int_{|x-y|>\epsilon} \phi(x-y) g(y)\dd\mathfrak{s}(y).
	\label{eqn:def_bound_op_s}
\end{equation}
%%%%%%%%%%%%%%%%%%%%%%%%%%%%%%%%%%%%%%%%%%%%%%%%%%%%%%%%
Thanks to \cite[Lemmas 3.3. \& 3.7.]{AMV14} we know that $C_\mathfrak{s}$ extends into a linear bounded self-adjoint operator on $L^2(\Sigma)^4$ and that for all $f\in L^2(\Sigma)^4$ we have the Plemelj-Sokhotski relations:
%%%%%%%%%%%%%%%%%%%%%%%%%%%%%%%%%%%%%%%%%%%%%%%%%%%%%%%%
\begin{equation}
	C_\pm(f) = \mp \frac{i}{2}(\alpha\cdot\bfn)f + C_\mathfrak{s}(f),\quad -4\big(C_\mathfrak{s}(\alpha\cdot\bfn)\big)^2(f) = f.
	\label{eqn:Plemelj}
\end{equation}
%%%%%%%%%%%%%%%%%%%%%%%%%%%%%%%%%%%%%%%%%%%%%%%%%%%%%%%%
In particular $C_\pm|_{L^2(\Sigma)^4}$ is a linear bounded operator from $L^2(\Sigma)^4$ onto itself, that we also denote $C_\pm$.

Now, we have all the tools to go through the proof of Proposition \ref{prop:propcaldproj}.
\begin{proof}[Proof of Proposition \ref{prop:propcaldproj}]
In this proof we use the notation introduced in the proposition. The proof is organised into four steps, each step corresponding to each point of Proposition \ref{prop:propcaldproj}. 

\myemph{Proof of {\it i)}} As operators from $L^2(\Sigma)^4$ onto itself, $C_\mp$ is the adjoint of $C_\pm$. Indeed, let $f,g\in L^2(\Sigma)^4$ we have:
\begin{align*}
	\ps{{C}_\pm(f)}{g}_{L^2(\Sigma)^4} 	&= \ps{\mp\frac12(i\alpha\cdot\bfn)f + C_\mathfrak{s}(f)}{g}_{L^2(\Sigma)^4}\\
						&=\ps{f}{\pm\frac12(i\alpha\cdot\bfn)g}_{L^2(\Sigma)^4} + \ps{f}{C_\mathfrak{s}(g)}_{L^2(\Sigma)^4}\\
						&= \ps{f}{C_\mp(g)}_{L^2(\Sigma)^4},
\end{align*}
where we used \cite[Lemmas 3.3. \& 3.7.]{AMV14}. Hence, by duality, if we consider $C_\pm$ as an operator from $H^{-1/2}(\Sigma)^4$ onto itself, its adjoint $C_\pm'$ is a linear bounded operator from $H^{1/2}(\Sigma)^4$ onto itself and acts as $C_\mp$. It yields
\[
	\mathcal{C}_\pm' = \big(\pm C_\pm(i\alpha\cdot\bfn)\big)' = \mp i (\alpha\cdot\bfn) C_\pm' = i (\alpha\cdot\bfn) C_\mp|_{H^{1/2}(\Sigma)^4} = \mathcal{C}_\pm^*|_{H^{1/2}(\Sigma)^4}.
\]

\myemph{Proof of {\it ii)}} As operators in $L^2(\Sigma)^4$, thanks to \eqref{eqn:Plemelj}, we have:
%%%%%%%%%%%%%%%%%%%%%%%%%%%%%%%%%%%%%%%%%%%%%%%%%%%%%%%%
\[
	\mathcal{C}_\pm = \pm i C_\pm\big((\alpha\cdot\bfn)\big) = \frac12 \pm i C_\mathfrak{s}\big((\alpha\cdot\bfn)\big).
\]
%%%%%%%%%%%%%%%%%%%%%%%%%%%%%%%%%%%%%%%%%%%%%%%%%%%%%%%%
Hence, \eqref{eqn:Plemelj} gives:
%%%%%%%%%%%%%%%%%%%%%%%%%%%%%%%%%%%%%%%%%%%%%%%%%%%%%%%%
\[
	\mathcal{C}_\pm^2 =  \frac14 - \big(C_\mathfrak{s}\big((\alpha\cdot\bfn)\big)\big)^2 \pm i C_\mathfrak{s}\big((\alpha\cdot\bfn)\big) =  \frac12 \pm i C_\mathfrak{s}\big((\alpha\cdot\bfn)\big) = \mathcal{C}_\pm.
\]
%%%%%%%%%%%%%%%%%%%%%%%%%%%%%%%%%%%%%%%%%%%%%%%%%%%%%%%%
Since for all $f\in\mathcal{C}^\infty(\Sigma)^4$ we have $\mathcal{C}_\pm^2(f) = \mathcal{C}_\pm(f)$, by density and continuity, this equality also holds in $H^{-1/2}(\Sigma)^4$. The proof of $(\mathcal{C}_\pm')^2 = \mathcal{C}_\pm'$ is handled similarly.

\myemph{Proof of {\it iii)}} Let $f\in\mathcal{C}^\infty(\Sigma)^4$. By definition and \eqref{eqn:Plemelj}, we have:
%%%%%%%%%%%%%%%%%%%%%%%%%%%%%%%%%%%%%%%%%%%%%%%%%%%%%%%%
\[
	\mathcal{C}_+(f) + \mathcal{C}_-(f) =  \frac12 f + i C_\mathfrak{s}\big((\alpha\cdot\bfn) f\big) + \frac12 f - i C_\mathfrak{s}\big((\alpha\cdot\bfn) f\big) = f.
\]
%%%%%%%%%%%%%%%%%%%%%%%%%%%%%%%%%%%%%%%%%%%%%%%%%%%%%%%%
As the last equation holds for any $f\in\mathcal{C}^\infty(\Sigma)^4$, by density and continuity, this is also true in $H^{-1/2}(\Sigma)^4$. Similarly, we obtain $\mathcal{C}_+^* + \mathcal{C}_-^* = {\rm Id}$.

\myemph{Proof of {\it iv)}} By definition and Point {\it i)}, we get:
%%%%%%%%%%%%%%%%%%%%%%%%%%%%%%%%%%%%%%%%%%%%%%%%%%%%%%%%
\[
	(\alpha\cdot\bfn)\mathcal{C}_\pm  = \pm i (\alpha\cdot\bfn)C_\pm(\alpha\cdot\bfn) = \mathcal{C}_\mp^*(\alpha\cdot\bfn).
\]
%%%%%%%%%%%%%%%%%%%%%%%%%%%%%%%%%%%%%%%%%%%%%%%%%%%%%%%%
Doing the composition with $(\alpha\cdot\bfn)$ on the left and on the right and using that $(\alpha\cdot\bfn)^2 = {\rm Id}$ we get the other identity:
%%%%%%%%%%%%%%%%%%%%%%%%%%%%%%%%%%%%%%%%%%%%%%%%%%%%%%%%
\[
	\mathcal{C}_\pm(\alpha\cdot\bfn) = (\alpha\cdot\bfn)\mathcal{C}_\mp^*.
\]
%%%%%%%%%%%%%%%%%%%%%%%%%%%%%%%%%%%%%%%%%%%%%%%%%%%%%%%%
\end{proof}

Now, we have all the tools to prove Proposition \ref{prop:regCaldtr}.
\begin{proof}[Proof of Proposition \ref{prop:regCaldtr}] Let $u\in H(\alpha,\Omega_\pm)$ and $u_n \in \mathcal{C}_0^\infty(\overline{\Omega}_\pm)$ that converges to $u$ in the ${\|\cdot\|_{H(\alpha,\Omega_\pm)}}$-norm. Let $f\in\mathcal{C}^\infty(\Sigma)^4$.

Let us start with $\mu\neq0$. Combining \eqref{eqn:Plemelj}, Corollary \ref{prop:Green_weak} and Proposition \ref{prop:nul_distrib} we have:
%%%%%%%%%%%%%%%%%%%%%%%%%%%%%%%%%%%%%%%%%%%%%%%%%%%%%%%%
\[
	\ps{\mathcal{C}_\mp(\mathfrak{t}_{\Sigma,\pm} u_n)}{f}_{L^2(\Sigma)^4} = \ps{\mathfrak{t}_{\Sigma,\pm} u_n}{\pm i (\alpha\cdot\bfn)C_\pm(f)}_{L^2(\Sigma)^4} = \ps{\mathcal{H}u_n}{\Phi(f)}_{L^2(\Omega_\pm)^4}.
\]
%%%%%%%%%%%%%%%%%%%%%%%%%%%%%%%%%%%%%%%%%%%%%%%%%%%%%%%%
Using the Cauchy-Schwarz inequality it gives:
%%%%%%%%%%%%%%%%%%%%%%%%%%%%%%%%%%%%%%%%%%%%%%%%%%%%%%%%
\begin{equation}
	\begin{array}{lcl}
		|\ps{\mathcal{C}_\mp(\mathfrak{t}_{\Sigma,\pm} u_n)}{f}_{L^2(\Sigma)^4}| 	&\leq& \|\mathcal{H}u_n\|_{L^2(\Omega_\pm)^4}\|\Phi(f)\|_{L^2(\Omega_\pm)^4}\\
												&\leq& c \|u_n\|_{H(\alpha,\Omega_\pm)}\|f\|_{H^{-1/2}(\Sigma)^4},
	\end{array}
\label{eqn:limit_reg_cald}
\end{equation}
%%%%%%%%%%%%%%%%%%%%%%%%%%%%%%%%%%%%%%%%%%%%%%%%%%%%%%%%
where the last inequality holds by Theorem \ref{thm:regu_distri}.
Hence, by density of $\mathcal{C}^\infty(\Sigma)^4$ in $H^{-1/2}(\Sigma)^4$, $\mathcal{C}_\mp(\mathfrak{t}_{\Sigma,\pm} u_n)$ defines a bounded linear form on $H^{-1/2}(\Sigma)^4$ and $\mathcal{C}_\mp(\mathfrak{t}_{\Sigma,\pm} u_n)\in H^{1/2}(\Sigma)^4$. Taking the limit $n\rightarrow +\infty$ in \eqref{eqn:limit_reg_cald}, it proves that $\mathcal{C}_\mp\circ\mathfrak{t}_{\Sigma,\pm}$ defines a bounded linear operator from $H(\alpha,\Omega_\pm)$ to $H^{1/2}(\Sigma)^4$.

Now, if $\mu=0$ and $\Omega_\pm$ is bounded the proof follows the exact same lines. Otherwise, we choose $R>0$ large enough such that
$\Sigma\subset B(0,R)$ and reproduce the proof with $\Omegat:=\Omega_\pm\cap B(3R)$ instead of $\Omega_\pm$ and $\chi \Phi(f)$ instead of $\Phi(f)$ where $\chi$ is a smooth bounded cut-off function such that $\chi(x) = 1$ for all $|x|<R$ and $\chi(x) = 0$ for all $|x|>2R$.	
\end{proof}

\begin{remark} Before going any further, we would like to point out that if $u\in L^2(\Omega)^4$ and is harmonic in $\Omega$, that is $u$ satisfies $\mathcal{H}(\mu)u = 0$, then there exist $c_1,c_2>0$ such that
\begin{equation}
	c_1 \|\tr u\|_{H^{-1/2}(\Sigma)^4} \leq \|u\|_{L^2(\Omega)^4} \leq c_2 \|\tr u\|_{H^{-1/2}(\Sigma)^4}.
	\label{eqn:equiv_norm}
\end{equation}
Roughly speaking, the norm in $L^2(\Omega)^4$ of a harmonic function is equivalent to the norm of its trace in $H^{-1/2}(\Sigma)^4$. Indeed, for $v\in L^2(\Omega)^4$,  we have
\begin{align*}
\ps{u}{v}_{L^2(\Omega)^4} = \ps{u}{\mathcal{H}(m)Vv}_{L^2(\Omega)^4} &= \ps{(i\alpha\cdot\bfn)\tr u}{(\tr \circ V) v}_{H^{-1/2}(\Sigma)^4,H^{1/2}(\Sigma)^4} \\&= \ps{\Phi\big((i\alpha\cdot\bfn)\tr u\big)}{v}_{L^2(\Omega)^4},
\end{align*}
where $V$ is the operator defined in Proposition \ref{prop:pot_regul}. It yields the reproducing formula ${u = \Phi\big((i\alpha\cdot\bfn)\tr u\big)}$ and then $\tr u = \mathcal{C}_+(\tr u)$. By Theorem \ref{thm:regu_distri}, there exists $c_2>0$ such that
\[
	\|u\|_{L^2(\Omega)^4} = \|u\|_{H(\alpha,\Omega)} \leq c_2 \|\tr u\|_{H^{-1/2}(\Sigma)^4}.
\]
Thanks to Proposition \ref{prop:transmweak}, there exists $c_1 > 0$ such that
\[
	c_1\|\tr u\|_{H^{-1/2}(\Sigma)^4} \leq \|u\|_{H(\alpha,\Omega)} = \|u\|_{L^2(\Omega)^4},
\]
which justifies Equation \eqref{eqn:equiv_norm}.
\end{remark}

\subsubsection{Regularisation of the anticommutator}
\label{subsubsec:reg_anticom}
This part deals with the proof of Proposition \ref{ref:reg_anti} but first, we need to introduce the next lemma.

\begin{lem} As operators in $L^2(\Sigma)^4$, the following equality holds
\[
\mathcal{A} = \{\alpha\cdot\bfn, C_{\mathfrak{s}}\} := (\alpha\cdot\bfn)C_\mathfrak{s} + C_\mathfrak{s}(\alpha\cdot\bfn).
\]
\label{lem:def_A}
\end{lem}

\begin{proof}[Proof of Lemma \ref{lem:def_A}] Let $f\in L^2(\Sigma)^4$, we have:
%%%%%%%%%%%%%%%%%%%%%%%%%%%%%%%%%%%%%%%%%%%%%%%%%%%%%%%%
\[
	\begin{array}{lcl}
		\mathcal{C}_\pm(f) - \mathcal{C}_\pm^*(f) 	&=& \pm i C_\pm\big((\alpha\cdot\bfn) f\big) \pm i (\alpha\cdot\bfn)C_\mp(f) \\
										&=&\pm i \Big(C_\pm\big((\alpha\cdot\bfn) f\big) + (\alpha\cdot\bfn)C_\mp(f)\Big).
	\end{array}
\]
%%%%%%%%%%%%%%%%%%%%%%%%%%%%%%%%%%%%%%%%%%%%%%%%%%%%%%%%
Thanks to \eqref{eqn:Plemelj}, last equation becomes:
%%%%%%%%%%%%%%%%%%%%%%%%%%%%%%%%%%%%%%%%%%%%%%%%%%%%%%%%
\[
	\begin{array}{lcl}
		\mathcal{C}_\pm(f) - \mathcal{C}_\pm^*(f) 	&=&\displaystyle \pm i \Big(\mp \frac{i}2 +ÊC_\mathfrak{s}\big((\alpha\cdot\bfn)f\big) \pm \frac{i}2 + (\alpha\cdot\bfn)C_\mathfrak{s}(f)\Big)\\
										&=& \pm i \Big(C_\mathfrak{s}\big((\alpha\cdot\bfn)f\big) + (\alpha\cdot\bfn)C_\mathfrak{s}(f)\Big)\\
										&=&\pm i \{C_\mathfrak{s},\alpha\cdot\bfn\}(f).
	\end{array}
\]
%%%%%%%%%%%%%%%%%%%%%%%%%%%%%%%%%%%%%%%%%%%%%%%%%%%%%%%%
By definition of $\mathcal{A}$ in \eqref{eqn:def_anticom}, it achieves the proof.
\end{proof}
Now, we have all the tools to go through the proof of Proposition \ref{ref:reg_anti}.

\begin{proof}[Proof of Proposition \ref{ref:reg_anti}] We prove that $\mathcal{A}$ is a bounded linear operator from $L^2(\Sigma)^4$ to $H^1(\Sigma)^4$. $\mathcal{A}$ being self-adjoint, Proposition \ref{ref:reg_anti} is obtained by duality and interpolation theory of the Sobolev spaces (see \cite[Part 2.1.7 \& Prop. 2.1.62.]{SS11}).

Remark that $\mathcal{A}$ is a singular integral operator with kernel
\[
	K(x,y) := (\alpha\cdot\bfn(x))\phi(x-y) + \phi(x-y)(\alpha\cdot\bfn(y))
\]
A simple algebraic computation yields
\[
	K(x,y) = \underset{:= K_1(x,y)}{\underbrace{2 (\bfn(x)\cdot\sfD)\psi(x-y)}} +  \underset{:= K_2(x,y)}{\underbrace{\phi(x-y)(\alpha\cdot(\bfn(y) - \bfn(x)))}}.
\]
$K_1$ is {\it a priori} a pseudo-homogeneous kernel of class $0$ (in the sense of \cite[\S 4.3.3]{Ned01}) but as $\Sigma$ is of class $\mathcal{C}^2$ it is actually pseudo-homogeneous of class $-1$. Indeed, we have
\[
	\psi(z) = \frac{e^{-|\mu z|}}{4\pi|z|} = \frac{1}{4\pi|z|} - \frac{|\mu|}{4\pi} + \frac{|\mu|^2|z|}{8\pi} + \dots,
\]
where the first term is homogeneous of class $-1$, the second one is smooth and more generally the $p$-th term is homogeneous of class $-(1+p)$. Now, remark that
\[
	K_1(x,y) = \frac1{4\pi}(\bfn(x)\cdot z)\big(-\frac1{|z|^3} - \frac{|\mu|^2}{8\pi|z|} + \dots\big).
\]
Thanks to \cite[Lemma 3.15]{Fol95}, we know that $\bfn(x)\cdot(x-y)$ behaves as $|x-y|^2$ when $x-y\rightarrow0$. Hence $K_1$ is pseudo-homogeneous of class $-1$ (see also \cite[\S 4.3.3 Ex. 4.5]{Ned01}). Thanks to \cite[Th. 4.3.2]{Ned01}, the singular integral operator of kernel $K_1$ is bounded from $L^2(\Sigma)^4$ to $H^1(\Sigma)^4$.

Now, remark that the study of $K_2$ reduces to the case $\mu = 0$. Indeed, we rewrite the kernel $K_2$ as:
\begin{equation}
	K_2(x,y) := \underset{:=K_{2,0}(x,y)}{\underbrace{\phi_{0}(x-y) (\alpha\cdot(\bfn(y) - \bfn(x)))}} + \underset{:= r(x,y)}{\underbrace{(\phi(x-y) - \phi_0(x-y))(\alpha\cdot(\bfn(y) - \bfn(x)))}}.
\label{eqn:ker_reg}
\end{equation}
We have
\[
	r(x,y) = \underset{:= r_1(x,y)}{\underbrace{(\alpha\cdot\sfD)(\psi_\mu - \psi_0)(x-y)\big(\alpha\cdot(\bfn(y) -\bfn(x))\big)}} + \underset{:=r_2(x,y)}{\underbrace{\mu \beta\psi_\mu(x-y)\big(\alpha\cdot(\bfn(y) -\bfn(x))\big)}}
\]
A computation yields that for all $x\in\mathbb{R}^3\setminus\{0\}$ we have
\[
	\psi_\mu(x) - \psi_0(x) = \frac1{4\pi}\sum_{n\geq0}\frac{|\mu|^{n+1} |x|^n}{(n+1)!},
\]
hence $\psi_\mu - \psi_0 \in \mathcal{C}^\infty(\Sigma)$ and, as $\Sigma$ is of class $\mathcal{C}^2$, $r_1(x,y)\in\mathcal{C}^1(\Sigma,\mathbb{C}^{4\times4})$. Thus, the integral operator of kernel $r_1$ is bounded from $L^2(\Sigma)^4$ to $H^1(\Sigma)^4$.

The kernel $r_2(x,y)$ rewrites
\[
	r_2(x,y) = \mu\sum_{j=1}^3\underset{:= r_{2,j}}{\underbrace{\psi_\mu(x-y)\big(\bfn_j(y) -\bfn_j(x)\big)}}\beta\alpha_j
\]
Remember that the single layer potential is pseudo-homogeneous of class $-1$ (see \cite[\S 4.3.3]{Ned01}). Moreover, as $\Sigma$ is of class $\mathcal{C}^2$, the multiplication by $\bfn_j$ is a bounded operator from $L^2(\Sigma)$ onto $L^2(\Sigma)$ and from $H^1(\Sigma)$ onto $H^1(\Sigma)$. Thus, the integral operator of kernel $r_{2,j}$ is bounded from $L^2(\Sigma)$ to $H^1(\Sigma)$.

The only thing left to prove is that the kernel $K_{2,0}$ introduced in \eqref{eqn:ker_reg} is bounded from $L^2(\Sigma)^4$ to $H^1(\Sigma)^4$.

The kernel $K_{2,0}$ can be rewritten as the sum of coefficients of the form:
\[
	c(x,y)\alpha_q\alpha_k,\text{with } c(x,y):=c_{q,k}(x,y) = i\frac{x_q - y_q}{4\pi|x-y|^3}(\bfn_k(y) - \bfn_k(x)),\quad q,k\in\{1,2,3\}.
\]
Consequently, the boundedness of $\mathcal{A}$ is equivalent to the one of the operators with kernels $c_{q,k}$.
Now, consider an atlas $(\Sigma_j,\Lambda_j)_{j\in\{1,N\}}$ covering the surface $\Sigma$, where $N\in\mathbb{N}^*$. By definition of an atlas, we have
\[
	\Sigma = \bigcup_{j\in\{1,\dots,N\}}\Sigma_j,
\]
and $\Lambda_j$ is a $\mathcal{C}^2$-diffeormorphism that maps $\Sigma_j$ to an open set $
\mathcal{U}_j := \Lambda_j(\Sigma_j)\subset\mathbb{R}^2.
$
We also introduce an adapted smooth partition of unity $(a_j)_{j\in\{1,\dots,N\}}$ such that
\[
	\sum_{j=1}^N a_j(x) = 1,\quad \text{for all } x\in\Sigma \text{ and } \mathsf{supp}(a_j) \subset \Sigma_j.
\]
Let us fix $j\in\{1,\dots,N\}$. For a function $f\in L^2(\Sigma)$ we decompose $f$ as:
\[
	\left\{
		\begin{array}{lcl}
			f & = & \displaystyle\sum_{j=1}^N a_j f,\\
			f_j & = & a_j f.
		\end{array}
	\right.
\]
Now, set $g(x) = \displaystyle\int_{y\in\Sigma} c(x,y) f(y) \dd\mathfrak{s}(y)$. We rewrite $g$ as
\[
		\left\{
		\begin{array}{lcl}
			g & = & \displaystyle\sum_{j=1}^N g^{[j]},\\
			g^{[j]}(x) & = & \displaystyle\int_{y\in\Sigma_j}c(x,y)f_j(y)\dd\mathfrak{s}(y).
		\end{array}
	\right.
\]
We only need to prove the regularity for $g^{[j]}$. Let $b_j$ be a smooth function such that
\[
	\mathsf{supp} (b_j)\subset \Sigma_j,\quad b_j(x) = 1 \text{ for all } x\in\mathsf{supp}(a_j).
\]
We introduce the function
\[
	g_j(x) := \int_{\Sigma_j} b_j(x)c(x,y)f_j(y)\dd\mathfrak{s}(y).
\]
We remark that
\[
	g^{[j]}(x) = g_j(x) + \int_{\Sigma_j}(1-b_j(x))c_{q,k}(x,y)f_j(y)\dd\mathfrak{s}(y),
\]
where the kernel in the last integral has no singularity in $x=y$ and is $\mathcal{C}^1$-smooth. Hence we only need to focus on $g_j$.
We perform the change of coordinates
\begin{equation}
	s = \Lambda_j(x),\quad t=\Lambda_j(y).
	\label{eqn:chg_var}
\end{equation}
We set $x(s):=\Lambda_j^{-1}(s)$ and $y(t):=\Lambda_j^{-1}(y)$. Hence we have:
\[
	g_j\big(x(s)\big) = \int_{\mathbb{R}^2}b_j\big(x(s)\big)c\big(x(s),y(t)\big)f_j\big(y(t)\big)J_j(t)\dd t,
\]
where $J_j$ is the Jacobian associated with the change of variables \eqref{eqn:chg_var}. This function of the variable $s$ has the same regularity as
\[
	h(s) := b_j\big(x(s)\big)\int_{\mathbb{R}^2}c\big(x(s),y(t)\big)\varphi(t)\dd t,
\]
where we set $\varphi(t) := f_j\big(y(t)\big)J_j(t)$. Note that $\varphi\in L^2(\mathbb{R}^2)$ and has compact support in $\mathcal{U}_j$.

Remark that
\begin{equation}
	b_j(x)c(x,y) = L(x - y)\big(b_j(y)\bfn_k(y) - b_j(x)\bfn_k(x)\big) + L(x-y)\big(b_j(x) - b_j(y)\big)\bfn_k(y),
\label{eqn:kern}
\end{equation}
where $L(z) := i\frac{z_q}{4\pi|z|^3}$.

As $\mathsf{supp}(\varphi)\subset \mathcal{U}_j$ and by definition of $b_j$, the second term in the right-hand side of \eqref{eqn:kern} reads, in local coordinates, as a kernel in $\mathcal{C}^1(\mathbb{R}^2)$. Thus as $\varphi$ has compact support this kernel regularises to $H^1(\mathbb{R}^2)$. Let us deal with the other term.

For $z_1,z_2\in\mathbb{R}^2$ such that $z_1\neq z_2$, $L$ expands as:
\[
	L(z_2) = L(z_1) + R(z_1,z_2),\text{ with } R (z_1,z_2) := \int_{0}^1(\nabla L)(z_1 + r (z_2-z_1))(z_2 - z_1)\dd r.
\]
Now, set $z_1 = {\sf d}(\Lambda_j^{-1})(s)(s-t)$ and $z_2 = x(s) - y(t) -z_1$, where ${\sf d}(\Lambda_j^{-1})(s) = J(s)$ is a jacobian. Thus, we have to take into account both kernels. Let us start with the first one. We have
\[
	L(z_1) := \frac{i}{4\pi}\frac{\big(J(s)(t-s)\big)_q}{|J(s)(t-s)|^3}.
\]
Remark that the chart $\Lambda_j$ can be choosen in such a way that $J(s)$ is an orthonormal matrix. We perform the change of variable
\[
	s' = J(s)s,\quad t' = J(s)t.
\]
$L(z_1)$ becomes a Riesz Kernel in $\mathbb{R}^2$ and the associated operator maps continuously $L^2(\mathbb{R}^2)$ onto itself (see \cite[Th. 1]{CZ57}). The singular integral operator with kernel
\[
	L\big(J(s)(t-s)\big)\Big(b_j\big(x(s)\big)\bfn_k\big(x(s)\big) - b_j\big(y(t)\big)\bfn_k\big(y(t)\big)\Big)
\]
can be seen as the commutator of the singular integral operator with kernel $L\big(J(s)(t-s)\big)$ and the $\mathcal{C}^1$-function $s\mapsto b_j\big(x(s)\big)\bfn_k\big(x(s)\big)$.
Hence, we recover the commutator of a Riesz kernel and a $\mathcal{C}^1$-smooth function.  Thanks to \cite{Cal65} we know that the commutator is bounded from $L^2(\mathbb{R}^2)$ to the usual homogeneous Sobolev space of order 1
\[
G^1(\mathbb{R}^2) :=\{f\in L^2(\mathbb{R}^2) : \int_{\mathbb{R}^2}|\xi||\mathcal{F}(f)(\xi)|^2\dd \xi\}.
\]
However, the commutator is bounded from $L^2(\mathbb{R}^2)$ onto itself because the multiplication is bounded on $L^2(\mathbb{R}^2)$. Thus, the first part regularises and we only have to take care of the remainder which is more regular.

Indeed, set $\tilde{R}(s,t) = R(z_1,z_2)\Big(\bfn_k\big(x(s)\big) - \bfn_k\big(y(t)\big)\Big)$. $\tilde{R}(s,t)\in\mathcal{C}^1(\mathbb{R}^2\setminus\{|s-t| = 0\})$ and, for some $C>0$, we have:
\[
	|\tilde{R}(s,t)\varphi(t)| \leq C |\varphi(t)|,\quad  |\partial_s \tilde{R}(s,t) \varphi(t)| \leq \frac{C}{|s-t|}|\varphi(t)|.
\]
As $\varphi(t)\in L^1(\mathbb{R}^2)$ because its support is compact, we obtain that the remainder also regularises with:
\[
\|\tilde{R}(\varphi)\|_{H^1(\mathbb{R}^2)} \leq C\|\varphi\|_{L^2(\mathbb{R}^2)},\quad\text{with }\tilde{R}(\varphi)(s) := \int_{\mathbb{R}^2}\tilde{R}(s,t)\varphi(t)\dd t.
\]
Consequently, $\mathcal{A}$ is a bounded operator from $L^2(\Sigma)^4$ to $H^1(\Sigma)^4$.
\end{proof}

\section{MIT bag model}
\label{sec:MIT-bag}
The MIT bag model was introduced by physicists in the MIT in order to understand quarks confinment, see \cite[\S 1.1.]{ALTR16} and the references therein to justify the pertinence of such a model. Mathematically, it is  defined as follows.
\begin{dfn}[MIT bag model] Let $m\in\mathbb{R}$. The MIT bag operator $\Big(\mathcal{H}_{\rm{MIT}}(m), \dom(\mathcal{H}_{\rm{MIT}}(m))\Big)$ is defined on the domain
\[
	\dom(\mathcal{H}_{\rm{MIT}}(m)) = \{u \in H^1(\Omega)^4 : \mathcal{B}\tr u = \tr u\ \text{on}\ \Sigma\},\quad \mathcal{B} = -i\beta (\alpha\cdot\bfn),
\]
by $\mathcal{H}_{\rm{MIT}}(m)u = \mathcal{H}(m)u$, for all $u\in \dom(\mathcal{H}_{\rm{MIT}}(m))$.
\end{dfn}

The following theorem is about the self-adjointness of the MIT bag operator. A similar result can be found in \cite[Thm. 1.5.]{ALTR16} and we state it here in order to illustrate our strategy to prove self-adjointness of Dirac operators. We also emphasize that it allows us to deal with $\mathcal{C}^2$ surfaces and with unbounded domains $\Omega$.

\begin{thm} The MIT bag operator $\big(\mathcal{H}_{\rm{MIT}}(m), \dom(\mathcal{H}_{\rm{MIT}}(m))\big)$ is self-adjoint.
\label{thm:MIT_sa}
\end{thm}

In Subsection \ref{subsec:symm_MIT} we prove that the MIT bag operator is symmetric. A description of its adjoint operator is given in Subsection \ref{subsec:desc_adj_MIT} while in Subsection \ref{subsec:sa_MIT} we conclude the proof of Theorem \ref{thm:MIT_sa}.
%%%%%%%%%%%%%%%%%%%%%%%%%%%%%%%%%%%%%%%%%%%%%%%%%%%%%%%%
\subsection{Symmetry of $\mathcal{H}_{\rm{MIT}}(m)$}
\label{subsec:symm_MIT}
%%%%%%%%%%%%%%%%%%%%%%%%%%%%%%%%%%%%%%%%%%%%%%%%%%%%%%%%
We prove the following proposition.
\begin{prop}The MIT bag operator $\big(\mathcal{H}_{\rm{MIT}}(m), \dom(\mathcal{H}_{\rm{MIT}}(m))\big)$ is symmetric.
\label{prop:sym_MIT}
\end{prop}

\begin{proof}[Proof of Proposition \ref{prop:sym_MIT}] Thanks to Green's formula in Lemma \ref{lem:Green_dir}, for any $u,v\in\dom(\mathcal{H}_{\rm{MIT}}(m))$ we have
\[
	\ps{\mathcal{H}(m) u}{v}_{L^2(\Omega)^4} = \ps{u}{\mathcal{H}(m) v}_{L^2(\Omega)^4} + \ps{(-i\alpha\cdot\bfn)\tr u}{\tr v}_{L^2(\Sigma)^4}.
\]
As $u\in\dom(\mathcal{H}_{\rm{MIT}}(m))$ we have $\tr u = \mathcal{B} \tr u$ thus
\begin{align*}
	\ps{(-i\alpha\cdot\bfn)\tr u}{\tr v}_{L^2(\Sigma)^4} &= \ps{(-i\alpha\cdot\bfn)\mathcal{B}\tr u}{\tr v}_{L^2(\Sigma)^4}\\
& = \ps{(-i\alpha\cdot\bfn)(-i\beta\alpha\cdot\bfn)\tr u}{\tr v}_{L^2(\Sigma)^4}\\
& =  \ps{\beta\tr u}{\tr v}_{L^2(\Sigma)^4}.
\end{align*}
Similarly, as $v\in\dom(\mathcal{H}_{\rm{MIT}}(m))$ we have $\tr v = \mathcal{B} \tr v$ thus
\begin{align*}
	\ps{(-i\alpha\cdot\bfn)\tr u}{\tr v}_{L^2(\Sigma)^4} &= \ps{(-i\alpha\cdot\bfn)\tr u}{\mathcal{B}\tr v}_{L^2(\Sigma)^4}\\
& = \ps{(-i\alpha\cdot\bfn)\tr u}{(-i\beta\alpha\cdot\bfn)\tr v}_{L^2(\Sigma)^4}\\
& =  -\ps{\beta\tr u}{\tr v}_{L^2(\Sigma)^4}.
\end{align*}
Hence we get
\[
	\ps{(-i\alpha\cdot\bfn)\tr u}{\tr v}_{L^2(\Sigma)^4} = - \ps{(-i\alpha\cdot\bfn)\tr u}{\tr v}_{L^2(\Sigma)^4} = 0,
\]
which concludes the proof.
\end{proof}

%%%%%%%%%%%%%%%%%%%%%%%%%%%%%%%%%%%%%%%%%%%%%%%%%%%%%%%%
\subsection{Description of the adjoint of $\mathcal{H}_{\rm{MIT}}(m)$}
\label{subsec:desc_adj_MIT}
%%%%%%%%%%%%%%%%%%%%%%%%%%%%%%%%%%%%%%%%%%%%%%%%%%%%%%%%
In this subsection we prove the following proposition.
\begin{prop}The following set equality holds.
\[
	\dom\Big(\big(\mathcal{H}_{\rm MIT}(m)\big)^*\Big) = \{u\in H(\alpha,\Omega) : \tr u = \mathcal{B}\tr u\},
\]
where the boundary condition has to be understood as an equality in $H^{-1/2}(\Sigma)^4$.
\label{prop:desc_adj_MIT}
\end{prop}

\begin{proof}[Proof of Proposition \ref{prop:desc_adj_MIT}] Let $\mathcal{V}$ denote the space on the right-hand side in Proposition \ref{prop:desc_adj_MIT}. We will prove the set equality proving each inclusion but first, recall that by definition
\[
	\dom\Big(\big(\mathcal{H}_{\rm MIT}(m)\big)^*\Big) = \left\{u \in L^2(\Omega)^4 :\begin{array}{c} \text{ there exists } w \in L^2(\Omega)^4 \text{ such that}\\\text{for all } v\in\dom\big(\mathcal{H}_{\rm MIT}(m)\big), \ps{u}{\mathcal{H}(m) v}_{L^2(\Omega)^4} = \ps{w}{v}_{L^2(\Omega)^4}\end{array}\right\}.
\]
\myemph{Inclusion $\mathcal{V} \subset \dom\Big(\big(\mathcal{H}_{\rm MIT}(m)\big)^*\Big)$.} Let $u\in\mathcal{V}$ and $v\in\dom\big(\mathcal{H}_{\rm MIT}(m)\big)$. Thanks to Corollary \ref{prop:Green_weak} we have
\[
	\ps{u}{\mathcal{H}(m) v}_{L^2(\Omega)^4} = \ps{\mathcal{H}(m) u}{v}_{L^2(\Omega)^4} +\ps{(i\alpha\cdot\bfn)\tr u}{\tr v}_{H^{-1/2}(\Sigma)^4,H^{1/2}(\Sigma)^4}.
\]
Now, as $\tr u = \mathcal{B}\tr u$ we have
\begin{align*}
	\ps{(-i\alpha\cdot\bfn)\tr u}{\tr v}_{H^{-1/2}(\Sigma)^4,H^{1/2}(\Sigma)^4} &= \ps{(-i\alpha\cdot\bfn)\mathcal{B}\tr u}{\tr v}_{H^{-1/2}(\Sigma)^4,H^{1/2}(\Sigma)^4}\\
& = \ps{(-i\alpha\cdot\bfn)(-i\beta\alpha\cdot\bfn)\tr u}{\tr v}_{H^{-1/2}(\Sigma)^4,H^{1/2}(\Sigma)^4}\\
& =  \ps{\beta\tr u}{\tr v}_{H^{-1/2}(\Sigma)^4,H^{1/2}(\Sigma)^4}.
\end{align*}
On the on the other hand as $\tr v = \mathcal{B}\tr v$ we have
\begin{align*}
	\ps{(-i\alpha\cdot\bfn)\tr u}{\tr v}_{H^{-1/2}(\Sigma)^4,H^{1/2}(\Sigma)^4} &= \ps{(-i\alpha\cdot\bfn)\tr u}{\mathcal{B}\tr v}_{H^{-1/2}(\Sigma)^4,H^{1/2}(\Sigma)^4}\\
& = \ps{(-i\alpha\cdot\bfn)\tr u}{(-i\beta\alpha\cdot\bfn)\tr v}_{H^{-1/2}(\Sigma)^4,H^{1/2}(\Sigma)^4}\\
& =  -\ps{\beta\tr u}{\tr v}_{H^{-1/2}(\Sigma)^4,H^{1/2}(\Sigma)^4}.
\end{align*}
Hence we get
\[
\ps{u}{\mathcal{H}(m) v}_{L^2(\Omega)^4} = \ps{\mathcal{H}(m) u}{v}_{L^2(\Omega)^4},
\]
which proves that $u\in\dom\Big(\big(\mathcal{H}_{\rm MIT}(m)\big)^*\Big)$.

\noindent
\myemph{Inclusion $\dom\Big(\big(\mathcal{H}_{\rm MIT}(m)\big)^*\Big)\subset \mathcal{V}$.} Let $u\in\dom\Big(\big(\mathcal{H}_{\rm MIT}(m)\big)^*\Big)$ and $v\in\mathcal{C}_0^\infty(\Omega)^4$. We have
\[
	\ps{\mathcal{H}(m)u}{v}_{\mathcal{D}'(\Omega)^4,\mathcal{D}(\Omega)^4}  = \ps{u}{\mathcal{H}(m)v}_{\mathcal{D}'(\Omega)^4,\mathcal{D}(\Omega)^4} = \ps{u}{\mathcal{H}(m)v}_{L^2(\Omega)^4}.
\]
As $u\in\dom\Big(\big(\mathcal{H}_{\rm MIT}(m)\big)^*\Big)$, there exists $w\in L^2(\Omega)^4$ such that
\[
\ps{u}{\mathcal{H}(m)v}_{L^2(\Omega)^4} = \ps{w}{v}_{L^2(\Omega)^4} = \ps{w}{v}_{\mathcal{D}'(\Omega)^4,\mathcal{D}(\Omega)^4}.
\]
As this is true for every $v\in\mathcal{C}_0^\infty(\Omega)^4$ we get $\mathcal{H}(m)u = w$ in $\mathcal{D}'(\Omega)^4$ and then in $L^2(\Omega)^4$. Thus we obtain $u\in H(\alpha,\Omega)$.
We introduce the matrices
\[
	P_+ = \frac12({\rm Id} + \mathcal{B}),\quad P_- = \frac12({\rm Id} - \mathcal{B}),
\]
they satisfy $\mathcal{B}P_+ = P_+$ and $\mathcal{B}P_- = -P_-$. Let $f\in H^{1/2}(\Sigma)^4$, we have $E(P_+f)\in\dom(\mathcal{H}_{\rm MIT}(m))$, where $E$ is the extension operator of Proposition \ref{prop:tr_th_cla}. As $u\in H(\alpha,\Omega))\cap \dom\Big(\big(\mathcal{H}_{\rm MIT}(m)\big)^*\Big)$ we have
\begin{align*}
	0 = \ps{\tr u}{(i\alpha\cdot\bfn)P_+ f}_{H^{-1/2}(\Sigma)^4,H^{1/2}(\Sigma)^4} &= \ps{\tr u}{(i\alpha\cdot\bfn)\mathcal{B}P_+ f}_{H^{-1/2}(\Sigma)^4,H^{1/2}(\Sigma)^4} \\&= -\ps{\mathcal{B}\tr u}{(i\alpha\cdot\bfn)P_+f}_{H^{-1/2}(\Sigma)^4,H^{1/2}(\Sigma)^4}.
\end{align*}
As $f = P_+f + P_-f$ we have
\begin{align*}
	\ps{\tr u}{(i\alpha\cdot\bfn) f}_{H^{-1/2}(\Sigma)^4,H^{1/2}(\Sigma)^4} &= \ps{\tr u}{(i\alpha\cdot\bfn)P_-f}_{H^{-1/2}(\Sigma)^4,H^{1/2}(\Sigma)^4} \\&= -\ps{\tr u}{(i\alpha\cdot\bfn)\mathcal{B}P_-f}_{H^{-1/2}(\Sigma)^4,H^{1/2}(\Sigma)^4}\\&=\ps{\mathcal{B}\tr u}{(i\alpha\cdot\bfn)P_-f}_{H^{-1/2}(\Sigma)^4,H^{1/2}(\Sigma)^4}\\
&=\ps{\mathcal{B}\tr u}{(i\alpha\cdot\bfn)f}_{H^{-1/2}(\Sigma)^4,H^{1/2}(\Sigma)^4}.
\end{align*}
As this is true for every $f\in H^{1/2}(\Sigma)^4$ we get that $\tr u = \mathcal{B}\tr u$. Thus $u\in\mathcal{V}$.
\end{proof}
%%%%%%%%%%%%%%%%%%%%%%%%%%%%%%%%%%%%%%%%%%%%%%%%%%%%%%%%
\subsection{Self-adjointness of the MIT bag model}
\label{subsec:sa_MIT}
%%%%%%%%%%%%%%%%%%%%%%%%%%%%%%%%%%%%%%%%%%%%%%%%%%%%%%%%
As in Subsection \ref{subsec:Cald_proj}, we set
%%%%%%%%%%%%%%%%%%%%%%%%%%%%%%%%%%%%%%%%%%%%%%%%%%%%%%%%%
\[
	\Omega_+:=\Omega\quad\text{and}\quad\Omega_-:=\mathbb{R}^3\setminus\overline{\Omega}.
\]
%%%%%%%%%%%%%%%%%%%%%%%%%%%%%%%%%%%%%%%%%%%%%%%%%%%%%%%%%
Let $\mu= 0$, we work with $\mathcal{C}_{\pm} = \mathcal{C}_{\pm,0}$ introduced in Definition \ref{dfn:cald_proj}.
Now, we have all the tools to prove Theorem \ref{thm:MIT_sa}.
\begin{proof}[Proof of Theorem \ref{thm:MIT_sa}] Let $u\in\dom\Big(\big(\mathcal{H}_{\rm MIT}(m)\big)^*\Big)$, thanks to Proposition \ref{prop:desc_adj_MIT} we know that
\[
	\tr u = \mathcal{B}\tr u.
\]
Moreover, thanks to Proposition \ref{prop:regCaldtr} we know that $\mathcal{C}_-(\tr u)\in H^{1/2}(\Sigma)^4$. Now, we prove that $\mathcal{C}_+(\tr u)\in H^{1/2}(\Sigma)^4$. Remark that for any $f\in H^{1/2}(\Sigma)^4$, $\mathcal{C}_{\pm}(\beta f) = -\beta\mathcal{C}_{\pm}(f)$. Thus, we have
\begin{multline}
	\mathcal{C}_{+}(\tr u) =  \mathcal{C}_{+}(\mathcal{B}\tr u) = i\beta\mathcal{C}_{+}\big((\alpha\cdot\bfn)\tr u\big) = i\beta (\alpha\cdot\bfn)\mathcal{C}_{-}^*(\tr u)\\ = i\beta(\alpha\cdot\bfn)\Big(\mathcal{C}_{-}(\tr u) + i\mathcal{A}(\tr u)\Big),
\label{eqn:eqn-samit}
\end{multline}
where we used {\it iv)} Proposition \ref{prop:propcaldproj} and Relation \eqref{eqn:def_anticom}. Thanks to Propositions \ref{prop:regCaldtr} and \ref{ref:reg_anti} the term in the right-hand side of \eqref{eqn:eqn-samit} is in $H^{1/2}(\Sigma)^4$ and thus
\[
	\tr u = \mathcal{C}_{+}(\tr u) + \mathcal{C}_{-}(\tr u) \in H^{1/2}(\Sigma)^4.
\]
Applying Proposition \ref{prop:regul_trac}, we obtain $u\in\dom(\mathcal{H}_{\rm MIT}(m))$. It proves the inclusion $\dom(\mathcal{H}_{\rm MIT}(m))\subset\dom\Big(\big(\mathcal{H}_{\rm MIT}(m)\big)^*\Big)$. The reciprocal inclusion is a consequence of Proposition \ref{prop:sym_MIT}.
\end{proof}
\section{Dirac operators coupled with electrostatic $\delta$-shell interactions}
\label{sec:Dir-shell}
Before stating the main result of this section, we need to introduce some notations and definitions.

As in Subsection \ref{subsec:Cald_proj} we set
\[
\Omega_+ := \Omega \text{ and } \Omega_- := \mathbb{R}^3\setminus\overline{\Omega}.
\]
We identify the space $L^2(\mathbb{R}^3)^4$ with $L^2(\Omega_+)^4\times L^2(\Omega_-)^4$ {\it via} the isomorphism
\begin{equation}
	\Lambda: u\in L^2(\mathbb{R}^3)^4 \mapsto (u_+,u_-) = (u|_{\Omega_+},u|_{\Omega_-}) \in L^2(\Omega_+)^4\times L^2(\Omega_-)^4,
\label{eqn:diff_lam}
\end{equation}
where $\Lambda^{-1}(u_1,u_2) := u_1 \mathds{1}_{\Omega_+} + u_2 \mathds{1}_{\Omega_-}$.

For $\tau\in\mathbb{R}$, we introduce the matrix valued function:
%%%%%%%%%%%%%%%%%%%%%%%%%%%%%%%%%%%%%%%%%%%%%%%%%%%%%%%%
\[
	 \mathcal{P}_\tau = \frac\tau2 + i(\alpha\cdot\bfn).
\] 
%%%%%%%%%%%%%%%%%%%%%%%%%%%%%%%%%%%%%%%%%%%%%%%%%%%%%%%%
For $(u_+,u_-)\in H^1(\Omega_+)^4\times H^1(\Omega_-)^4 $ we define the following transmission condition in $H^{1/2}(\Sigma)^4$
%%%%%%%%%%%%%%%%%%%%%%%%%%%%%%%%%%%%%%%%%%%%%%%%%%%%%%%%
\begin{equation}
	\mathcal{P}_\tau \mathfrak{t}_{\Sigma,+} u_+ + \mathcal{P}_\tau^*\mathfrak{t}_{\Sigma,-} u_- = 0, \quad \text{on } \Sigma.
\label{eqn:transm_cond0}
\end{equation}
%%%%%%%%%%%%%%%%%%%%%%%%%%%%%%%%%%%%%%%%%%%%%%%%%%%%%%%%
Alternatively, as $\mathcal{P_\tau}$ is invertible, we can see the transmission condition as
\begin{equation}
	\mathfrak{t}_{\Sigma,+} u_+ = \mathcal{R}_\tau \mathfrak{t}_{\Sigma,-} u_-, \text{ with } \mathcal{R}_\tau := \frac{1}{\tau^2/4 + 1}\Big(1 -\frac{\tau^2}{4} + \tau (i\alpha\cdot\bfn)\Big).
\label{eqn:transm_cond0_alt}
\end{equation}
This transmission condition is the natural one generated by an electrostatic $\delta$-interaction of strength $\tau$ supported on $\Sigma$, this is discussed further on in Subsection \ref{subsec:rmk_transm}. As there is no confusion possible, from now on, $\tr u_\pm$ denotes $\mathfrak{t}_{\Sigma,\pm} u_\pm$.

Now, let us define the operator we are interested in.
\begin{dfn} Let $\tau\in\mathbb{R}$ and $m\in\mathbb{R}$. The Dirac operator coupled with an electrostatic $\delta$-shell interaction of strength $\tau$ is the operator $\Big(\mathcal{H}_\tau(m),\dom \big(\mathcal{H}_\tau(m) \big)\Big)$, acting on $L^2(\mathbb{R}^3)^4$ and defined on the domain
%%%%%%%%%%%%%%%%%%%%%%%%%%%%%%%%%%%%%%%%%%%%%%%%%%%%%%%%
\begin{equation}
	\dom\big(\mathcal{H}_\tau(m)\big) = \big\{(u_+,u_-)\in H^1(\Omega_+)^4\times H^1(\Omega_-)^4 : (u_+,u_-)Ê\text{ satisfies } \eqref{eqn:transm_cond0}\big\}.
	\label{eqn:domdir}
\end{equation}
%%%%%%%%%%%%%%%%%%%%%%%%%%%%%%%%%%%%%%%%%%%%%%%%%%%%%%%%
It acts in the sense of distributions as $\mathcal{H}_\tau(m) u = (\mathcal{H}(m) u_+,\mathcal{H}(m) u_-)$ where we identify an element of $L^2(\Omega_+)^4\times L^2(\Omega_-)^4$ with an element of $L^2(\mathbb{R}^3)^4$ {\it via} \eqref{eqn:diff_lam}.
\end{dfn}
Note that $\mathcal{P}_\tau$ and $i\alpha\cdot\bfn$ commute, that is
%%%%%%%%%%%%%%%%%%%%%%%%%%%%%%%%%%%%%%%%%%%%%%%%%%%%%%%%
\begin{equation}
	\mathcal{P}_\tau (i\alpha\cdot\bfn) = (i\alpha\cdot\bfn)\mathcal{P}_{\tau}.
		\label{eqn:pte_P_tau0}
\end{equation}
%%%%%%%%%%%%%%%%%%%%%%%%%%%%%%%%%%%%%%%%%%%%%%%%%%%%%%%%
Finally, if $\tau=0$, we recover the usual free Dirac operator $\mathcal{H}_0(m)$ with domain $\dom\big(\mathcal{H}_0(m)\big) = H^1(\mathbb{R}^3)^4$.

We can now state the main result of this section.
\begin{thm}
Let $m\in\mathbb{R}$. The following holds:
\begin{itemize}
	\item[i)] If $\tau\neq\pm2$, the operator $\Big(\mathcal{H}_\tau(m),\dom \big(\mathcal{H}_\tau(m) \big)\Big)$ is self-adjoint.
	\item[ii)] If $\tau = \pm 2$, the operator $\Big(\mathcal{H}_\tau(m),\dom \big(\mathcal{H}_\tau(m) \big)\Big)$ is essentially self-adjoint and we have
%%%%%%%%%%%%%%%%%%%%%%%%%%%%%%%%%%%%%%%%%%%%%%%%%%%%%%%%
\[
	\dom ({\mathcal{H}}_\tau(m))\subsetneq\dom (\overline{\mathcal{H}}_\tau(m)) = \{(u_+,u_-)\in H(\alpha,\Omega_+)\times H(\alpha,\Omega_-) : (u_+,u_-) \text{ satisfies } \eqref{eqn:transm_cond0}\},
\]
%%%%%%%%%%%%%%%%%%%%%%%%%%%%%%%%%%%%%%%%%%%%%%%%%%%%%%%
where the transmission condition holds in $H^{-1/2}(\Sigma)^4$.
\end{itemize}
\label{th:sadirac}
\end{thm}
In \cite[Thm. 3.8]{AMV14}, the authors are able to prove the self-adjointness of the operator under the condition $\tau\neq\pm2$. However, except in the particular case $\Sigma=\mathbb{R}^2\times\{0\}$, they do not provide a description when $\tau=\pm2$. The proof, of Theorem \ref{th:sadirac} differs significantly from what is done in \cite{AMV14} and follows the philosophy of \cite{BFSVDB16} with the use of Calder\'on projectors. In particular, it allows us to understand the specific case $\tau=\pm 2$.	

%%%%%%%%%%%%%%%%%%%%%%%%%%%%%%%%%%%%%%%%%%%%%%%%%%%%%%%%%%%%%%%%%%%%%%%%%%%%%%%%%%%%%%%%%%%%%%%%%%%%%%%%%%%%%%%%%%%%%%%%%%%%%%%%%%%%%%%%%%%%%%%%%%%%%%%%%%%%%%%%%%%%%%%%%%%%%%%%%%%%%%%%%%%%%%%%%%%%%%%%%%%%%%%%%%%%%%%%%%%%%%%%%%%%%%%%%%%%%%%%%%%%%%%%%%%%%%%%%%%%%%%%%%%%%%%%%%%%%%%%%%%%%%%%%%%%%%%%%%%%%%%%%%%%%%%%%%%%%%%%%%%%%%%%%%%%%%%%%%%%%%%%%%%%%%%%%%%%%%%%%%%%%%%%%%%%%%%%%%%%%%%%%%%%%%%%%%%%%%%%%%%%%%%%%%%%%%%%%%%%%%%%%%%%%%%%%%%%%%%%%%%%%%%%%%%%%%%%%%%%%%%%%%%%%%%%%%%%%%%%%%%%%%%%%%%%%%%%%%%%%%%%

\subsection{Remarks on the transmission condition}
\label{subsec:rmk_transm}
This subsection aims to justify the expression of Transmission condition \eqref{eqn:transm_cond0}.
Our goal is to define the operator that formally writes
\[
	\mathcal{H}_\tau(m) = \mathcal{H}(m) + \tau \delta_\Sigma,
\]
where, for $u\in H^1(\Omega_+)^4\times H^1(\Omega_-)^4$, $\delta_\Sigma u$ is the distribution defined as
\[
	\ps{\delta_\Sigma u}{v}_{\mathcal{D}'(\mathbb{R}^3)^4,\mathcal{D}(\mathbb{R}^3)^4} := \frac12\int_{\Sigma}\ps{\tr u_+(x) + \tr u_-(x)}{v(x)}_{\mathbb{C}^4}\dd\mathfrak{s}(x),\quad\text{for all } v\in\mathcal{C}_0^\infty(\mathbb{R}^3)^4.
\]

We are interested in functions $u\in L^2(\mathbb{R}^3)^4$ such that
%%%%%%%%%%%%%%%%%%%%%%%%%%%%%%%%%%%%%%%%%%%%%%%%%%%%%%%%
\[
	(\mathcal{H}(m) + \tau\delta_\Sigma(x)\rm{Id})u \in L^2(\mathbb{R}^3)^4.
\]
%%%%%%%%%%%%%%%%%%%%%%%%%%%%%%%%%%%%%%%%%%%%%%%%%%%%%%%%
For example, if $u=(u_+,u_-)\in H^1(\Omega_+)^4\times H^1(\Omega_-)^4$, a computation in the sense of distributions yields
%%%%%%%%%%%%%%%%%%%%%%%%%%%%%%%%%%%%%%%%%%%%%%%%%%%%%%%%
\[
	\begin{array}{lcl}
		(\mathcal{H}(m) + \tau\delta_\Sigma(x)\rm{Id})u		& = & \displaystyle \alpha\cdot \sfD u + m\beta u +\frac\tau2(\tr u_+ + \tr u_-)\delta_\Sigma\\
		& = & \displaystyle \{(\alpha\cdot\sfD) u\} + m\beta u - i\alpha\cdot\bfn(\tr u_- - \tr u_+)\delta_\Sigma +\frac\tau2(\tr u_+ + \tr u_-)\delta_\Sigma\\
		& = & \displaystyle\underset{\in L^2(\mathbb{R}^3)^4}{\underbrace{\{(\alpha\cdot\sfD) u\} + m\beta u}} + \big(\frac{\tau}{2}(\tr u_+ + \tr u_-) -i\alpha\cdot\bfn(\tr u_- - \tr u_+)\big)\delta_\Sigma,
	\end{array}
\]
%%%%%%%%%%%%%%%%%%%%%%%%%%%%%%%%%%%%%%%%%%%%%%%%%%%%%%%%
where we set $\{(\alpha\cdot\sfD) u\}|_{\Omega_\pm} = (\alpha\cdot\sfD)u_\pm$. Now, we would like the last term in the right-hand side to be zero. It reads:
%%%%%%%%%%%%%%%%%%%%%%%%%%%%%%%%%%%%%%%%%%%%%%%%%%%%%%%%
\[
	(\frac{\tau}2\rm{Id} + i\alpha\cdot\bfn) \tr u_+ + (\frac{\tau}{2}\rm{Id} - i\alpha\cdot\bfn)\tr u_- =0.
\]
In particular, it justifies that for $u\in\dom \big(\mathcal{H}_\tau(m) \big)$, $\mathcal{H}_\tau(m)u\in L^2(\mathbb{R}^3)^4$.

\begin{remark} Let $\varepsilon=\pm1$ and $\tau=2\varepsilon$. Let $u=(u_+,u_-)\in\dom\big(\mathcal{H}_\tau(m)\big)$, $u_\pm$ can be rewritten $u_\pm=(u_\pm^{[1]},u_\pm^{[2]})$ and, for $x\in \Sigma$, the transmission condition reads
%%%%%%%%%%%%%%%%%%%%%%%%%%%%%%%%%%%%%%%%%%%%%%%%%%%%%%%%
\begin{equation}
	\left(\begin{array}{c}
		u_+^{[1]}(x)\\
		u_+^{[2]}(x)
	\end{array}\right) =
	\left(\begin{array}{cc}
		0 & -i\varepsilon \sigma\cdot\bfn(x)\\
		 -i\varepsilon \sigma\cdot\bfn(x) & 0
	\end{array}\right)\left(\begin{array}{c}
		u_-^{[1]}(x)\\
		u_-^{[2]}(x)
	\end{array}\right) = \left(\begin{array}{c}
		-i\varepsilon \sigma\cdot\bfn u_-^{[2]}(x)\\
		-i\varepsilon \sigma\cdot\bfn u_-^{[1]}(x)
	\end{array}\right).
	\label{eqn:cond_transm1}
\end{equation}
%%%%%%%%%%%%%%%%%%%%%%%%%%%%%%%%%%%%%%%%%%%%%%%%%%%%%%%%
The specificity of $\tau=\pm2$ lies in the fact that the system uncouples: $u_+^{[1]}$, respectively $u_+^{[2]}$, only "sees" $u_-^{[2]}$, respectively $u_-^{[1]}$.
\end{remark}

\subsection{Symmetry of $\mathcal{H}_\tau (m)$}  We prove the following proposition.
\begin{prop} The Dirac operator coupled with an electrostatic $\delta$-interaction $\big(\mathcal{H}_\tau (m),\dom(\mathcal{H}_\tau (m))\big)$ is a symmetric operator.
\label{prop:sym_op}
\end{prop}
\begin{proof}[Proof of Proposition \ref{prop:sym_op}] Let $u=(u_+,u_-),v=(v_+,v_-) \in H^1(\Omega_+)^4\times H^1(\Omega_-)^4$. Green's formula of Lemma \ref{lem:Green_dir} yields
%%%%%%%%%%%%%%%%%%%%%%%%%%%%%%%%%%%%%%%%%%%%%%%%%%%%%%%%
\begin{align*}
	\ps{\mathcal{H}_\tau(m) u}{v}_{L^2(\mathbb{R}^3)^4} & = \ps{\mathcal{H}(m) u_+}{v_+}_{L^2(\Omega_+)^4} + \ps{\mathcal{H}(m) u_-}{v_-}_{L^2(\Omega_-)^4}\\
& = \ps{u_+}{\mathcal{H}(m)v_+}_{L^2(\Omega_+)^4} + \ps{u_-}{\mathcal{H}(m)v_-}_{L^2(\Omega_-)^4}\\
& + \ps{(-i\alpha\cdot\bfn)\tr u_+}{\tr v_+}_{L^2(\Sigma)^4} - \ps{(-i\alpha\cdot\bfn)\tr u_-}{\tr v_-}_{L^2(\Sigma)^4},\\
\end{align*}
%%%%%%%%%%%%%%%%%%%%%%%%%%%%%%%%%%%%%%%%%%%%%%%%%%%%%%%%
which rewrites
\[
	\ps{\mathcal{H}_\tau(m) u}{v}_{L^2(\mathbb{R}^3)^4} - \ps{u}{\mathcal{H}_\tau(m)v}_{L^2(\mathbb{R}^3)^4} = \ps{(-i\alpha\cdot\bfn)\tr u_+}{\tr v_+}_{L^2(\Sigma)^4} - \ps{(-i\alpha\cdot\bfn)\tr u_-}{\tr v_-}_{L^2(\Sigma)^4}.
\]
%%%%%%%%%%%%%%%%%%%%%%%%%%%%%%%%%%%%%%%%%%%%%%%%%%%%%%%%
Now, assume that both $u$ and $v$ satisfy Transmission condition \eqref{eqn:transm_cond0_alt}, we have
%%%%%%%%%%%%%%%%%%%%%%%%%%%%%%%%%%%%%%%%%%%%%%%%%%%%%%%%
\begin{align*}
	\ps{\mathcal{H}_\tau(m) u}{v}_{L^2(\mathbb{R}^3)^4} - \ps{u}{\mathcal{H}_\tau(m)v}_{L^2(\mathbb{R}^3)^4} & = \ps{(-i\alpha\cdot\bfn)\tr u_+}{\tr v_+}_{L^2(\Sigma)^4} - \ps{(-i\alpha\cdot\bfn)\tr u_-}{\tr v_-}_{L^2(\Sigma)^4} \\ &= \ps{(-i\alpha\cdot\bfn)\mathcal{R}_\tau \tr u_-}{\mathcal{R}_\tau \tr v_-}_{L^2(\Sigma)^4} - \ps{(-i\alpha\cdot\bfn)\tr u_-}{\tr v_-}_{L^2(\Sigma)^4}\\&=\ps{(\mathcal{R}_\tau^*(-i\alpha\cdot\bfn)\mathcal{R}_\tau + i\alpha\cdot\bfn) \tr u_-}{\tr v_-}_{L^2(\Sigma)^4}.
\end{align*}
%%%%%%%%%%%%%%%%%%%%%%%%%%%%%%%%%%%%%%%%%%%%%%%%%%%%%%%%
By definition of $\mathcal{R}_\tau$ we have:
%%%%%%%%%%%%%%%%%%%%%%%%%%%%%%%%%%%%%%%%%%%%%%%%%%%%%%%%
\[
	\mathcal{R}_\tau^*(-i\alpha\cdot\bfn)\mathcal{R}_\tau + i\alpha\cdot\bfn = 0.
\]
%%%%%%%%%%%%%%%%%%%%%%%%%%%%%%%%%%%%%%%%%%%%%%%%%%%%%%%%
It achieves the proof of Proposition \ref{prop:sym_op}.
\end{proof}

\subsection{Domain of the adjoint} The aim of this subsection is to prove the following result.
\begin{prop}
We have
%%%%%%%%%%%%%%%%%%%%%%%%%%%%%%%%%%%%%%%%%%%%%%%%%%%%%%%%
\[
	\dom(\mathcal{H}_\tau(m)^*) = \{(u_+,u_-)\in H(\alpha,\Omega_+)\times H(\alpha,\Omega_-) : (u_+,u_-) \text{ satisfies  \eqref{eqn:transm_cond0} in } H^{-1/2}(\Sigma)^4\}.
\]
%%%%%%%%%%%%%%%%%%%%%%%%%%%%%%%%%%%%%%%%%%%%%%%%%%%%%%%%
\label{prop:eg_dom_adjoint}
\end{prop}
\begin{proof}[Proof of Proposition \ref{prop:eg_dom_adjoint}] Let $\mathcal{V}$ be the set on the right-hand side in Proposition \ref{prop:eg_dom_adjoint}. We prove this result proving each inclusion.

\noindent
\myemph{Inclusion $\mathcal{V}\subset \dom(\mathcal{H}_\tau(m)^*)$.} Let $u=(u_+,u_-)\in\mathcal{V}$ and $v=(v_+,v_-)\in\dom(\mathcal{H}_\tau(m))$. Thanks to Corollary \ref{prop:Green_weak} we have
%%%%%%%%%%%%%%%%%%%%%%%%%%%%%%%%%%%%%%%%%%%%%%%%%%%%%%%
\begin{align*}
	\ps{u}{\mathcal{H}_\tau(m) v}_{L^2(\mathbb{R}^3)^4} &= \ps{\mathcal{H}(m) u_+}{v_+}_{L^2(\Omega_+)^4} + \ps{\mathcal{H}(m)u_-}{v_-}_{L^2(\Omega_-)^4}\\
&+ \ps{(-i\alpha\cdot\bfn)\tr u_+}{\tr v_+}_{H^{-1/2}(\Sigma)^4,H^{1/2}(\Sigma)^4} - \ps{(-i\alpha\cdot\bfn)\tr u_-}{\tr v_-}_{H^{-1/2}(\Sigma)^4,H^{1/2}(\Sigma)^4}\\
&=\ps{\{\mathcal{H}(m) u\}}{v}_{L^2(\mathbb{R}^3)^4} + \ps{(\mathcal{R}_\tau^*(-i\alpha\cdot\bfn)\mathcal{R}_\tau + i\alpha\cdot\bfn) \tr u_-}{\tr v_-}_{H^{-1/2}(\Sigma)^4,H^{1/2}(\Sigma)^4},
\end{align*}
%%%%%%%%%%%%%%%%%%%%%%%%%%%%%%%%%%%%%%%%%%%%%%%%%%%%%%%
where $\{\mathcal{H}(m)u\} = (\mathcal{H}(m)u_+)\mathds{1}_{\Omega_+} + (\mathcal{H}(m)u_-)\mathds{1}_{\Omega_-} \in L^2(\mathbb{R}^3)^4$. As in the proof of Proposition \ref{prop:sym_op} we remark that
\[
	\mathcal{R}_\tau^*(-i\alpha\cdot\bfn)\mathcal{R}_\tau + i\alpha\cdot\bfn = 0,
\]
thus
\[
	\ps{u}{\mathcal{H}_\tau(m) v}_{L^2(\mathbb{R}^3)^4} = \ps{\{\mathcal{H}(m) u\}}{v}_{L^2(\mathbb{R}^3)^4},
\]
which proves that $u\in\dom(\mathcal{H}_\tau(m)^*)$.

\noindent
\myemph{Inclusion $\dom(\mathcal{H}_\tau(m)^*)\subset\mathcal{V}$.} Let $u=(u_+,u_-)\in \dom(\mathcal{H}_\tau(m)^*)$ and $v=(v_+,v_-)\in \mathcal{C}_0^\infty(\Omega_+)^4\times\mathcal{C}_0^\infty(\Omega_-)^4$. We have
\[
	\ps{\mathcal{H}(m) u }{v}_{\mathcal{D}'(\mathbb{R}^3)^4,\mathcal{D}(\mathbb{R}^3)^4} = \ps{u }{\mathcal{H}(m) v}_{\mathcal{D}'(\mathbb{R}^3)^4,\mathcal{D}(\mathbb{R}^3)^4} = \ps{u_+}{\mathcal{H}(m) v_+}_{L^2(\Omega_+)} + \ps{u_-}{\mathcal{H}(m) v_-}_{L^2(\Omega_+)}.
\]
As $u\in\dom(\mathcal{H}_\tau(m)^*)$, there exists $w\in L^2(\mathbb{R}^3)^4$ such that
\[
	\ps{u_+}{\mathcal{H}(m) v_+}_{L^2(\Omega_+)} + \ps{u_-}{\mathcal{H}(m) v_-}_{L^2(\Omega_+)} = \ps{w}{v}_{L^2(\mathbb{R}^3)^4} = \ps{w_+}{v_+}_{\mathcal{D}'(\Omega_+)^4,\mathcal{D}(\Omega_+)^4} + \ps{w_-}{v_-}_{\mathcal{D}'(\Omega_-)^4,\mathcal{D}(\Omega_-)^4},
\]
where $w_\pm = w \mathds{1}_{\Omega_\pm}$. As this is true for every $v\in\mathcal{C}_0^\infty(\Omega_+)^4\times\mathcal{C}_0^\infty(\Omega_-)^4$ we get $\mathcal{H}(m) u_\pm = w_\pm$ in $\mathcal{D}'(\Omega_\pm)^4$ and then in $L^2(\Omega_\pm)^4$. Thus $u\in H(\alpha,\Omega_-)\times H(\alpha,\Omega_+)$.

Let $f\in H^{1/2}(\Sigma)^4$, we have:
\begin{multline*}
	\ps{\tr u_+}{(i\alpha\cdot\bfn)f}_{H^{-1/2}(\Sigma)^4,H^{1/2}(\Sigma)^4} - \ps{\tr u_-}{(i\alpha\cdot\bfn)f}_{H^{-1/2}(\Sigma)^4,H^{1/2}(\Sigma)^4} \\
 = \ps{\tr u_+}{(i\alpha\cdot\bfn)\mathcal{P}_\tau^*f}_{H^{-1/2}(\Sigma)^4,H^{1/2}(\Sigma)^4} - \ps{\tr u_-}{(i\alpha\cdot\bfn)(-\mathcal{P}_\tau)f}_{H^{-1/2}(\Sigma)^4,H^{1/2}(\Sigma)^4}\\
+\ps{\tr u_+}{(i\alpha\cdot\bfn)({\rm Id} - \mathcal{P}_\tau^*)f}_{H^{-1/2}(\Sigma)^4,H^{1/2}(\Sigma)^4} - \ps{\tr u_-}{(i\alpha\cdot\bfn)({\rm Id} +\mathcal{P}_\tau)f}_{H^{-1/2}(\Sigma)^4,H^{1/2}(\Sigma)^4}
\end{multline*}
Now, we remark that the function $\big(E_+ (\mathcal{P}_\tau^*f), -E_- (\mathcal{P}_\tau f)\big)\in\dom\big(H_\tau(m)\big)$, where $E_\pm$ is the extension operator of Proposition \ref{prop:tr_th_cla} in $H^1(\Omega_\pm)$. Thus we get
\[
\ps{\tr u_+}{(i\alpha\cdot\bfn)\mathcal{P}_\tau^*f}_{H^{-1/2}(\Sigma)^4,H^{1/2}(\Sigma)^4} - \ps{\tr u_-}{(i\alpha\cdot\bfn)(-\mathcal{P}_\tau)f}_{H^{-1/2}(\Sigma)^4,H^{1/2}(\Sigma)^4} = 0,
\]
which implies
\[
	\ps{\tr u_+}{(i\alpha\cdot\bfn)\mathcal{P}_\tau^*f}_{H^{-1/2}(\Sigma)^4,H^{1/2}(\Sigma)^4} + \ps{\tr u_-}{(i\alpha\cdot\bfn)\mathcal{P}_\tau f}_{H^{-1/2}(\Sigma)^4,H^{1/2}(\Sigma)^4} = 0.
\]
Hence, for all $f\in H^{1/2}(\Sigma)^4$, we get
\[
\ps{\mathcal{P}_\tau \tr u_+ + \mathcal{P}_\tau^* \tr u_-}{(i\alpha\cdot\bfn)f}_{H^{-1/2}(\Sigma)^4,H^{1/2}(\Sigma)^4} = 0,
\]
which proves that $u$ satisfies Transmission condition \eqref{eqn:transm_cond0} in $H^{-1/2}(\Sigma)^4$.
\end{proof}

\subsection{Self-adjointness}
The aim of this subsection is to prove {\it i)}~Theorem \ref{th:sadirac}.

\begin{proof}[Proof of {\it i)}~Theorem \ref{th:sadirac}] Let $u\in\dom\big(\mathcal{H}_\tau(m)^*\big)$, thanks to Proposition \ref{prop:eg_dom_adjoint} we know that
\begin{equation}
	\mathcal{P}_\tau\tr u_+ + \mathcal{P}_\tau^* \tr u_- = 0.
\label{eqn:cond_trans_preu}
\end{equation}
Thanks to {\it iii)} Proposition \ref{prop:propcaldproj} we have
\begin{align*}
	\eqref{eqn:cond_trans_preu}&\Longleftrightarrow\left\{\begin{array}{lcl}\mathcal{C}_+\big(\mathcal{P}_\tau\tr u_+ + \mathcal{P}_\tau^* \tr u_-\big) & = & 0\\\mathcal{C}_-\big(\mathcal{P}_\tau\tr u_+ + \mathcal{P}_\tau^* \tr u_-\big) & = & 0\end{array}\right.,\\
&\Longleftrightarrow \left\{	\begin{array}{lcl}
									\frac\tau2\big(\mathcal{C}_+(\tr u_+) + \mathcal{C}_+(\tr u_-)\big) + i(\alpha\cdot\bfn)\mathcal{C}_-^*(\tr u_+ - \tr u_-) &=& 0\\
\frac\tau2\big(\mathcal{C}_-(\tr u_+) + \mathcal{C}_-(\tr u_-)\big) + i(\alpha\cdot\bfn)\mathcal{C}_+^*(\tr u_+ - \tr u_-) &=& 0
								\end{array}\right.,\\
&\Longleftrightarrow\left\{	\begin{array}{lcl}
									\frac\tau2\big(\mathcal{C}_+(\tr u_+) + \mathcal{C}_+(\tr u_-)\big) + i(\alpha\cdot\bfn)\big(\mathcal{C}_-(\tr u_+) - \mathcal{C}_-(\tr u_-) + i\mathcal{A}(\tr u_+ - \tr u_-)\big) &=& 0\\
\frac\tau2\big(\mathcal{C}_-(\tr u_+) + \mathcal{C}_-(\tr u_-)\big) + i(\alpha\cdot\bfn)\big(\mathcal{C}_+(\tr u_+) - \mathcal{C}_+(\tr u_-) - i\mathcal{A}(\tr u_+ - \tr u_-)\big) &=& 0
								\end{array}\right.,
\end{align*}
where we also used {\it iv)} Proposition \ref{prop:propcaldproj} and Relation \eqref{eqn:def_anticom}. This system rewrites as
\begin{multline}
	\left(\begin{array}{cc}\frac\tau2&-i\alpha\cdot\bfn\\i\alpha\cdot\bfn&\frac\tau2\end{array}\right)\left(\begin{array}{c}\mathcal{C}_+(\tr u_+)\\\mathcal{C}_-(\tr u_-)\end{array}\right) = \left(\begin{array}{cc}-\frac\tau2&-i\alpha\cdot\bfn\\i\alpha\cdot\bfn&-\frac\tau2\end{array}\right)\left(\begin{array}{c}\mathcal{C}_+(\tr u_-)\\\mathcal{C}_-(\tr u_+)\end{array}\right) \\+ \left(\begin{array}{c}(\alpha\cdot\bfn)\mathcal{A}(\tr u_+ - \tr u_-)\\-(\alpha\cdot\bfn)\mathcal{A}(\tr u_+ - \tr u_-)\end{array}\right).
\label{eqn:sadelta}
\end{multline}
Now, thanks to Propositions \ref{prop:regCaldtr} and \ref{ref:reg_anti}, the right-hand side is in $H^{1/2}(\Sigma)^8$ and the matrix in the left-hand side is invertible in $H^{1/2}(\Sigma)^8$ as long as $\tau\neq\pm2$. Thus $\tr u_\pm\in H^{1/2}(\Sigma)^4$ and applying Proposition \ref{prop:regul_trac} we obtain the inclusion $\dom\big(\mathcal{H}_\tau(m)^*\big)\subset\dom\big(\mathcal{H}_\tau(m)\big)$. The reciprocal inclusion is a consequence of Proposition \ref{prop:sym_op}.
\end{proof}
\subsection{Essential self-adjointness when $\tau = \pm 2$} Now, we prove {\it ii)}~Theorem \ref{th:sadirac}. All along this subsection, we set $\varepsilon=\pm1$ and let $\tau=2\varepsilon$. We work with the operator $\Phi_\pm:=\Phi_{\Omega_\pm,\mu}$ defined in \eqref{eqn:dfn_laypot} with a fixed $\mu\neq0$.

We have the following proposition.
\begin{prop} Let $\tau=2\varepsilon$. The following holds:
\[
	\overline{\mathcal{H}_\tau(m)} = \mathcal{H}_\tau^*(m).
\]
In particular, $\overline{\mathcal{H}_\tau(m)}$ is self-adjoint.
\label{prop:clos_dom}
\end{prop}
For $u\in\dom(\mathcal{H}_\tau(m)^*)$, Transmission condition \eqref{eqn:transm_cond0_alt} simply reads
\begin{equation}
	\tr u_+ = i\varepsilon(\alpha\cdot \bfn) \tr u_-,
	\label{eqn:transm_cond}
\end{equation}
as an equality in $H^{-1/2}(\Sigma)^4$.

Let us introduce a few notation. For $u=(u_+,u_-)\in\dom(\mathcal{H}_\tau(m)^*)$, if $(f_{n})_{n\in\mathbb{N}}$ is a sequence of functions $\mathcal{C}^\infty(\Sigma)^4$ that converges to $\tr u_-$ in the $\|\cdot\|_{H^{-1/2}(\Sigma)^4}$-norm, we introduce:
\begin{equation}
	\left\{\begin{array}{rcl}
	u_{n,-} &=& u_- + i\Phi_-\big((\alpha\cdot \bfn) (\tr u_- - f_n)\big),\\
	u_{n,+} &=& u_+ - \varepsilon\Phi_+(f_n -\tr u_-) +\varepsilon E_+\Big(\mathcal{A}\big((\alpha\cdot \bfn)(f_n - \tr u_-)\big)\Big),
	\end{array}\right.
	\label{eqn:def_seq}
\end{equation}
where $E_+$ is the extension operator defined in Proposition \ref{prop:tr_th_cla} with $\Omega=\Omega_+$ and $\mathcal{A}:=\mathcal{A}_\mu$ is the anticommutator \eqref{eqn:def_anticom}. We have the following lemma:
\begin{lem} Let $u=(u_+,u_-)\in\dom(\mathcal{H}_\tau(m)^*)$ and $(f_{n})_{n\in\mathbb{N}}$ be a sequence of functions $\mathcal{C}^\infty(\Sigma)^4$ that converges to $\tr u_-$ in the $\|\cdot\|_{H^{-1/2}(\Sigma)^4}$-norm. If $u_n=(u_{n,-},u_{n,+})$ is the sequence defined in \eqref{eqn:def_seq} then:
\begin{itemize}
	\item[i)] $u_n\in H^1(\Omega_+)^4\times H^1(\Omega_-)^4$,
	\item[ii)] $(u_{n,+},u_{n,-})$ satisfies Transmission condition \eqref{eqn:transm_cond} in $H^{1/2}(\Sigma)^4$,
	\item[iii)] $u_n$ converges to $u$ in the $\|\cdot\|_{H(\alpha,\Omega_+)\times H(\alpha,\Omega_-)}$-norm.
\end{itemize}
\label{lem:prop_seq_clos}
\end{lem}
We postpone the proof of Lemma \ref{lem:prop_seq_clos} until the end of this subsection. We now have all the tools to prove Proposition \ref{prop:clos_dom}.
\begin{proof}[Proof of Proposition \ref{prop:clos_dom}] As $\mathcal{H}_\tau(m)$ is symmetric it is closable and we have $\overline{\mathcal{H}_\tau(m)}\subset\mathcal{H}_\tau(m)^*$. Now we deal with the other inclusion.

Let $u=(u_+,u_-)\in\dom(\mathcal{H}_\tau(m)^*)$ and $(f_{n})_{n\in\mathbb{N}}$ be a sequence of functions $\mathcal{C}^\infty(\Sigma)^4$ that converges to $\tr u_-$ in the $\|\cdot\|_{H^{-1/2}(\Sigma,\mathbb{C}^4)}$-norm. Let $u_n = (u_{n,-},u_{n,+})$ be as in \eqref{eqn:def_seq}. Thanks to {\it i)-ii)}~Lemma \ref{lem:prop_seq_clos}, we know that $u_n\in\dom(\mathcal{H}_\tau(m))$. Moreover thanks to {\it iii)}~Lemma \ref{lem:prop_seq_clos} we know that $\|u_n - u\|_{H(\alpha,\Omega_+)\times H(\alpha,\Omega_-)}\rightarrow0$ when $n\rightarrow+\infty$. Consequently, $u\in\dom(\overline{\mathcal{H}_\tau(m)})$ and we obtain the reversed inclusion, that is $\mathcal{H}_\tau(m)^*\subset\overline{\mathcal{H}_\tau(m)}$.
\end{proof}

\begin{proof}[Proof of Lemma \ref{lem:prop_seq_clos}] For the sake of clarity, this proof is split into two steps. The proofs of {\it i)} and {\it ii)} are gathered in Step 1. Step 2 deals with the proof of {\it iii)}.

Let $u=(u_+,u_-)\in\dom(\mathcal{H}_\tau(m)^*)$ and $(u_n)_{n\in\mathbb{N}}$ be the associated sequence defined in \eqref{eqn:def_seq}.

\myemph{Step 1.} By definition, $u_n\in H(\alpha,\Omega_+)\times H(\alpha,\Omega_-)$. Thanks to Proposition \ref{prop:regul_trac}, it is enough to prove that the traces $\tr u_{n,-}$ and $\tr u_{n,+}$ are in $H^{1/2}(\Sigma)^4$. Let us start with $u_{n,-}$, we have:
\[
	\tr u_{n,-} = \tr u_-  -  \mathcal{C}_-(\tr u_-) + \mathcal{C}_-(f_n) = \mathcal{C}_+(\tr u_-) + \mathcal{C}_-(f_n).
\]
by Proposition \ref{prop:regCaldtr}, the first term in the right hand-side is in $H^{1/2}(\Sigma)^4$. The second term is also in $H^{1/2}(\Sigma)^4$ by Corollary \ref{thm:cont_tr}.

Let us prove that Transmission condition \eqref{eqn:transm_cond} holds. Taking into account {\it iii)}~Proposition \ref{prop:propcaldproj} and \eqref{eqn:def_anticom} we get:
\[
	\begin{array}{lcl}
	i\varepsilon (\alpha\cdot \bfn) \tr u_{n,-} 	& = & i\varepsilon (\alpha\cdot \bfn) \big(\tr u_- - \mathcal{C}_-(\tr u_- - f_n)\big)\\
										& = & \tr u_+ - i\varepsilon (\alpha\cdot \bfn)\mathcal{C}_-(\tr u_- - f_n)\\
										& = & \tr u_+ - i\varepsilon \mathcal{C}_+^*\big((\alpha\cdot \bfn) (\tr u_- - f_n)\big)\\
										& = & \tr u_+ - i\varepsilon\Big(\mathcal{C}_+\big((\alpha\cdot \bfn) (\tr u_- - f_n)\big) - i\mathcal{A}\big((\alpha\cdot \bfn) (\tr u_- - f_n)\big)\Big)\\
										& = &\tr u_+ - i\varepsilon (\alpha\cdot \bfn) \mathcal{C}_-^*\big((\tr u_- - f_n)\big) - \varepsilon\mathcal{A}\big((\alpha\cdot \bfn) (\tr u_- - f_n)\big)\\
										& = & \tr u_{n,+}.
	\end{array}
\]
As $u_{n,-}\in H^1(\Omega_-)^4$, it implies $u_{n,+}\in H^1(\Omega_+)^4$ and we get {\it ii)}.

\myemph{Step 2.} In this step, we prove {\it iii)}.

\noindent
Let us start with $u_{n,-}$, we have:
\begin{equation}
	u_{n,-} - u_- = i \Phi_-\big((\alpha\cdot \bfn)(\tr u_- - f_n)\big).
	\label{eqn:diff_u_neg}
\end{equation}
Hence, by Theorem \ref{thm:regu_distri} there exists a constant $c_1>0$ such that:
\[
	\|u_{n,-} - u_-\|_{H(\alpha,\Omega_-)} \leq c_1 \|\tr u_- - f_n\|_{H^{-1/2}(\Sigma)^4}.
\]
By hypothesis, the term in the right-hand side goes to zero as $n$ goes to infinity so we obtain $u_{n,-} \underset{n\rightarrow+\infty}{\longrightarrow} u_-$ in the $\|\cdot\|_{H(\alpha,\Omega_-)}$-norm. 

Now we deal with $u_{n,+}$. We have:
\begin{equation}
	u_{n,+} - u_{+} = - \varepsilon \Phi_+(f_n -\tr u_-) + \varepsilon E_+\Big(\mathcal{A}\big(\alpha\cdot N (f_n - \tr u_-)\big)\Big).
	\label{eqn:diff_u_pos}
\end{equation}
It yields
\begin{equation}
	\|u_{n,+} - u_+\|_{H(\alpha,\Omega_+)} \leq \varepsilon \|\Phi_+(f_n -\tr u_-)\|_{H(\alpha,\Omega_+)} + \varepsilon\Big\|E_+\Big(\mathcal{A}\big(\alpha\cdot N (f_n - \tr u_-)\big)\Big)\Big\|_{H(\alpha,\Omega_+)}.
	\label{eqn:up_bound1}
\end{equation}
By Theorem \ref{thm:regu_distri}, there exists $c_2>0$ such that:
\begin{equation}
	\|\Phi(f_n -\tr u_-)\|_{H(\alpha,\Omega_+)} \leq c_2 \|\tr u_- - f_n\|_{H^{-1/2}(\Sigma)^4}.
	\label{eqn:eqn1_lemm}
\end{equation}
Thanks to Proposition \ref{prop:tr_th_cla} and Proposition \ref{ref:reg_anti}, there exists $c_3>0$ such that:
\begin{equation}
	\begin{array}{lcl}
	\Big\|E_+\Big(\mathcal{A}\big((\alpha\cdot \bfn) (f_n - \tr u_-)\big)\Big)\Big\|_{H(\alpha,\Omega_+)} &\leq& \Big\|E_+\Big(\mathcal{A}\big((\alpha\cdot \bfn) (f_n - \tr u_-)\big)\Big)\Big\|_{H^1(\Omega_+)^4}\\
	&\leq& c_3 \|\tr u_- - f_n\|_{H^{-1/2}(\Sigma)^4}.
	\end{array}
	\label{eqn:eqn2_lemm}
\end{equation}
The upper-bounds of Equations \eqref{eqn:eqn1_lemm} and \eqref{eqn:eqn2_lemm} combined with \eqref{eqn:up_bound1} yield the existence of $c>0$ such that:
\[
	\|u_{n,+} - u_+\|_{H(\alpha,\Omega_+)} \leq c \|\tr u_- - f_n\|_{H^{-1/2}(\Sigma)^4}.
\]
By hypothesis, the right-hand side converges to zero as $n$ goes to infinity and so, we get the convergence of $u_{n,+}$ to $u_+$ in the $\|\cdot\|_{H(\alpha,\Omega_+)}$-norm.
\end{proof}

We finish this subsection proving Theorem \ref{th:sadirac}.
\begin{proof}[Proof of Theorem \ref{th:sadirac}] The only thing left to prove is that $\dom ({\mathcal{H}}_\tau(m))\subsetneq\dom (\overline{\mathcal{H}}_\tau(m))$. Let $0\neq f\in H^{-1/2}(\Sigma)^4$ such that $f\notin H^{1/2}(\Sigma)^4$. Either $\mathcal{C}_+(f)$ of $\mathcal{C}_-(f)$ does not belong to $H^{1/2}(\Sigma)$. Assume $\mathcal{C}_-(f) \notin H^{1/2}(\Sigma)^4$, we set $g = \mathcal{C}_-(f)$. We consider the function
\[
u = (u_+,u_-) = \Bigg(\varepsilon \Phi_+(g) - \varepsilon E_+\Big(\mathcal{A}\big((\alpha\cdot\bfn)g\big)\Big),\Phi_-\big((i\alpha\cdot\bfn)g\big)\Bigg).
\]
By definition, $u\in H(\alpha,\Omega_+)\times H(\alpha,\Omega_-)$ and we have:
\begin{align*}
	i\varepsilon(\alpha\cdot\bfn)\tr u_- &= -i\varepsilon(\alpha\cdot\bfn)\mathcal{C}_-(g)\\
&= -\varepsilon i \mathcal{C}_+\big((\alpha\cdot\bfn)g\big) - \varepsilon\mathcal{A}\big((\alpha\cdot\bfn)g\big)\\
& = \tr u_+.
\end{align*}
Hence $u$ satisfies Transmission condition \eqref{eqn:transm_cond} which gives $u\in\dom\big(\overline{\mathcal{H}_\tau(m)}\big)$ by Proposition \ref{prop:clos_dom}. However, $u\notin\dom(\mathcal{H}_\tau(m))$, otherwise $\tr u_- \in H^{1/2}(\Sigma)^4$ which is not possible because $\tr u_- = -g = - \mathcal{C}_-(f) \notin H^{1/2}(\Sigma)^4$.

If $\mathcal{C}_+(f)\notin H^{1/2}(\Sigma)^4$ the proof goes along the same line setting $g = \mathcal{C}_+(f)$ and considering the function
\[
	u = (u_+,u_-) = \Bigg(\Phi_+\big((i\alpha\cdot\bfn)g\big), - \varepsilon \Phi_-(g) + \varepsilon E_-\Big(\mathcal{A}\big((\alpha\cdot\bfn)g\big)\Big)\Bigg),
\]
where $E_-$ is the extension operator of Proposition \ref{prop:tr_th_cla} for the domain $\Omega_-$.
\end{proof}

\begin{remark} When $\tau=\pm 2$, the domain of the extension differs significantly from the one of the initial operator. Indeed, following the proof of Theorem \ref{th:sadirac}, we remark that any function $0\neq f \in H^{-1/2}(\Sigma)^4$ that is not in $H^{1/2}(\Sigma)^4$ generate an element of $\dom (\overline{\mathcal{H}}_\tau(m))$ that is not in $\dom ({\mathcal{H}}_\tau(m))$.
\end{remark}

\myemph{Acknowledgements.}
Part of this project was written while T.~O.-B. was working at the BCAM-Basque Center for Applied Mathematics and  was supported by the Basque Government through the BERC 2014-2017 program and by Spanish Ministry of Economy and Competitiveness MINECO: BCAM Severo Ochoa excellence accreditation SEV-2013-0323. Now, T.~O.-B. is supported by a public grant as part of the {\it Investissement d'avenir} project, reference ANR-11-LABX-0056-LMH, LabEx LMH.

L.V. is supported by an ERCEA Advanced Grant 2014 669689 - HADE, by the MINECO project MTM2014-53850-P and MINECO Severo Ochoa excellence accreditation SEV-2013-0323

\bibliographystyle{abbrv}

\end{document}

%% file: commands.tex
\newcommand{\beq}{\begin{equation}}
\newcommand{\eeq}{\end{equation}}

\newcommand{\comm}[1]{}

\newcommand{\tr}{\mathfrak{t}_\Sigma}

\makeatletter

\@addtoreset{equation}{section}

\def\@listI{
    \leftmargin\leftmargini
    \parsep 1.5pt plus 1pt minus 1pt
    \topsep -1.5pt plus 1pt minus 1pt
    \itemsep \parsep}
\let\@listi\@listI
\makeatother

\usepackage{color}

\newcommand{\ps}[2]{\left\langle#1,#2\right\rangle}

\newcommand\sfD{\mathsf{D}}

\newcommand\myemph[1]{\textbf{\emph{#1}}} 

\definecolor{darkred}{rgb}{0.5,0.1,0.1}

\newcommand\dd{{\mathsf{d}}}

\def\softness{0.4}

\definecolor{softred}{rgb}{1,\softness,\softness}
\definecolor{softgreen}{rgb}{\softness,1,\softness}
\definecolor{softblue}{rgb}{\softness,\softness,1}
\definecolor{softrg}{rgb}{1,1,\softness}
\definecolor{softrb}{rgb}{1,\softness,1}
\definecolor{softgb}{rgb}{\softness,1,1}

\newcounter{counter_a}

\usepackage[latin1]{inputenc}
\usepackage[T1]{fontenc}

\numberwithin{figure}{section}
\numberwithin{equation}{section}
\theoremstyle{plain}% default
\newtheorem*{thm*}{Theorem}
\newtheorem{thm}{Theorem}[section]

\newtheorem{lem}[thm]{Lemma}
\newtheorem{prop}[thm]{Proposition}

\newtheorem{corollary}[thm]{Corollary}

\newtheorem{cor}[thm]{Corollary}

\newtheorem{dfn}[thm]{Definition}
\theoremstyle{remark}
\newtheorem{remark}[thm]{Remark}
\theoremstyle{plain}

\newcommand{\bfn}{\mathbf{n}}

\newcommand{\beu}{\begin{equation*}}
\newcommand{\eeu}{\end{equation*}}
\newcommand{\besu}{\begin{equation*}
\begin{aligned}}
\newcommand{\eesu}{\end{aligned}
\end{equation*}}
\newcommand{\bes}{\begin{equation}
\begin{aligned}}
\newcommand{\ees}{\end{aligned}
\end{equation}}

\newcommand\CC{\mathbb C}

\newcommand\void[1]{}

\newcommand{\dom}{\mathrm{dom}\,}

\newcommand{\Omegat}{\widetilde{\Omega}}